\documentclass{lmcs}
\pdfoutput=1

\usepackage{lastpage}
\lmcsdoi{19}{4}{24}
\lmcsheading{}{\pageref{LastPage}}{}{}%
{Sep.~12,~2022}{Dec.~15,~2023}{}


\usepackage{hyperref}
\usepackage{cleveref}
\usepackage[utf8]{inputenc}
\usepackage{amssymb,amsfonts,amsmath,amsthm}
\usepackage{mathtools}
\usepackage{nicefrac}
\usepackage{booktabs}
\usepackage{adjustbox}
\usepackage{bm}
\usepackage{wrapfig}
\usepackage{soul}
\usepackage{makecell}
\usepackage{tikz}
\usetikzlibrary{cd,decorations.pathreplacing}
\usetikzlibrary{arrows,decorations.pathmorphing,positioning,fit}
\usetikzlibrary{trees,shapes}
\usetikzlibrary {automata, positioning}
\usetikzlibrary{tikzmark}

\newenvironment{highlight}{%
    \begin{center}%
    \medskip%
    \begin{minipage}{0.865\textwidth}%
    \emph%
    \bgroup%
}{%
    \egroup%
    \end{minipage}%
    \medskip%
    \end{center}%
}

\newcommand{\Runs}{Runs}

\newcommand{\V}[2]{V_{#1}^{(#2)}}

\newcommand{\qeddef}{\hfill$\triangle$}

\tikzstyle{callState} = [state,rectangle,minimum size=8mm]
\tikzstyle{retState} = [state,diamond,minimum size=9mm]
\tikzstyle{myautomaton} = [shorten >=1pt, >=stealth',auto,initial text=] 
\tikzstyle{myBox} = [state,rectangle]

\newcommand{\transStackThree}[3]{#1\\[-2pt]#2\\[-2pt]#3}

\newcommand{\Nats}{\mathbb{N}_0}
\newcommand{\Rats}{\mathbb{Q}}
\newcommand{\Dist}{\mathit{Dist}}
\newcommand{\condProb}[2]{\PP(#1 \mid #2)}
\renewcommand{\prob}{p} 

\newcommand{\stack}{\gamma}
\newcommand{\existsInfty}{\overset{\infty}{\exists}}
\newcommand{\partition}[3]{#1 \uplus #2 \uplus #3}
\newcommand{\set}[1]{\{#1\}}
\newcommand{\Set}[2]{\{\,#1 \,\mid\,#2\,\}}

\newcommand{\automProd}[2]{#1 \times #2}
\newcommand{\pstate}[2]{(#1,#2)} 
\newcommand{\psymb}[2]{\langle#1,#2\rangle} 

\newcommand{\npda}{\mathcal{A}}
\newcommand{\npdaStates}{S} 
\newcommand{\npdaTrans}{\delta}
\newcommand{\abStack}{\Gamma}
\newcommand{\abStackNoBot}{\abStack_{\text{-}\botStack}}

\newcommand{\sInit}{s_0}
\newcommand{\sInitNpda}{s_0}
\newcommand{\botStack}{\bot}
\newcommand{\npdaInitCustom}[6]{(#1, \, #2, \, #3, \, #4, \, #5, \, #6)}
\newcommand{\npdaInit}{\npdaInitCustom{\npdaStates}{\sInit}{\abStack}{\botStack}{\npdaTrans}{\ab}}
\newcommand{\run}{\rho}

\newcommand{\dpda}{\mathcal{D}} 
\newcommand{\dvpa}{\mathcal{D}} 
\newcommand{\transCall}[4]{#1 \xrightarrow{#2} #3#4}
\newcommand{\transCallAutom}[5]{#1 \xrightarrow{#2}_{#3} #4#5}
\newcommand{\transInt}[3]{#1 \xrightarrow{#2} #3}
\newcommand{\transIntAutom}[4]{#1 \xrightarrow{#2}_{#3} #4}
\newcommand{\transRet}[4]{#1#2 \xrightarrow{#3} #4}
\newcommand{\transRetAutom}[5]{#1#2 \xrightarrow{#3}_{#4} #5}

\newcommand{\buchiset}{F}
\newcommand{\prioFun}{\Omega}
\newcommand{\prioval}{k} 
\newcommand{\accSet}[1]{\mathit{Parity}_{#1}} 

\newcommand{\ppda}{\Delta}
\newcommand{\ppdaInit}{(\ppdaStates, \, \sInitPpda, \, \abStack, \, \botStack, \, \ppdaTransFun, \, \ab, \, \labeling)}
\newcommand{\ppdaStates}{Q}
\newcommand{\ppdaStatesCall}{\ppdaStates_{\callSubscript}}
\newcommand{\ppdaStatesRet}{\ppdaStates_{\retSubscript}}
\newcommand{\ppdaStatesInt}{\ppdaStates_{\intSubscript}}
\newcommand{\ppdaTransFun}{P}
\newcommand{\Pint}{\ppdaTransFun_{\intSubscript}}

\newcommand{\Pret}{\ppdaTransFun_{\retSubscript}}
\newcommand{\ppdaTrans}[3]{#1 \xrightarrow{#2} #3}
\newcommand{\ppdaTransCall}[4]{\ppdaTrans{#1}{#2}{#3#4}}
\newcommand{\ppdaTransInt}[3]{\ppdaTrans{#1}{#2}{#3}}
\newcommand{\ppdaTransRet}[4]{\ppdaTrans{#1#2}{#3}{#4}}
\newcommand{\sInitPpda}{q_0}

\newcommand{\labeling}{\lambda}

\newcommand{\AP}{{AP}}

\newcommand{\PP}{\mathbb{P}}

\newcommand{\nonTermProb}[1]{[#1{\uparrow}]}
\newcommand{\termProbFromTo}[2]{[#1{\downarrow}#2]}
\newcommand{\ppdaMcSemantic}[1]{\mathcal{D}_{#1}}

\newcommand{\ab}{\Sigma}
\newcommand{\callSubscript}{\mathsf{call}}
\newcommand{\retSubscript}{\mathsf{ret}}
\newcommand{\intSubscript}{\mathsf{int}}
\newcommand{\visTypeMetaSubscript}{\mathsf{type}}
\newcommand{\abCall}{\ab_{\callSubscript}}
\newcommand{\abRet}{\ab_{\retSubscript}}
\newcommand{\abInt}{\ab_{\intSubscript}}

\newcommand{\word}{w}

\newcommand{\sh}[1]{sh(#1)} 

\newcommand{\lang}[1]{\mathcal{L}(#1)} 
\newcommand{\langConcat}[2]{#1.#2}

\newcommand{\stepi}[2]{\mathit{step}_{#1}(#2)} 

\newcommand{\wordRestr}[2]{#1{\downarrow}_{#2}}
\newcommand{\stepRestr}[1]{\wordRestr{#1}{\mathit{Steps}}}
\newcommand{\footprint}[1]{#1 {\Downarrow} _{\mathit{Steps}}}

\newcommand{\stateAtPos}[2]{#1@#2}
\newcommand{\stepAtPos}[1]{step@#1}
\newcommand{\stackAt}[1]{sh@#1}
\newcommand{\stepMC}[1]{\mathcal{M}_{#1}}

\newcommand{\globally}{\square}
\newcommand{\finally}{\lozenge}
\renewcommand{\next}{\bigcirc}
\newcommand{\gnext}{\next^g}
\newcommand{\anext}{\next^a}
\newcommand{\caller}{\next^-}
\newcommand{\until}{\mathcal{U}}
\newcommand{\symboluntil}[1]{\mathcal{U}^{#1}}
\newcommand{\guntil}{\symboluntil{g}}
\newcommand{\auntil}{\symboluntil{a}}
\newcommand{\calleruntil}{\symboluntil{-}}
\newcommand{\symbolsucc}[3]{\mathit{succ}^{#1}_{#2}(#3)}
\newcommand{\asucc}[2]{\symbolsucc{a}{#1}{#2}}
\newcommand{\callersucc}[2]{\symbolsucc{-}{#1}{#2}}
\newcommand{\matchingReturn}[2]{\mathit{MR}_{#1}(#2)}
\newcommand{\undefined}{\ensuremath{\mathit{undef}}}
\newcommand{\typeset}{\set{\callSubscript, \intSubscript, \retSubscript}}

\newcommand{\PSPACE}{\mathsf{PSPACE}}
\newcommand{\PTIME}{\mathsf{P}}
\newcommand{\EXPTIME}{\mathsf{EXPTIME}}
\newcommand{\TwoEXPTIME}{\mathsf{2EXPTIME}}
\newcommand{\EXPSPACE}{\mathsf{EXPSPACE}}
\newcommand{\TwoEXPSPACE}{\mathsf{2EXPSPACE}}

\renewcommand{\O}{\mathcal{O}}

\newcommand{\bitmorespace}[1]{\,#1\,}
\newcommand{\morespace}[1]{~#1~}
\newcommand{\waymorespace}[1]{\quad #1 \quad}

\newcommand{\lland}{\morespace{\land}}
\newcommand{\qland}{\waymorespace{\land}}

\newcommand{\llor}{\morespace{\lor}}

\newcommand{\bmid}{\bitmorespace{\mid}}

\newcommand{\ccoloneqq}{\morespace{\coloneqq}}

\newcommand{\eeq}{\morespace{=}}
\newcommand{\qeq}{\waymorespace{=}}

\newenvironment{theorem}{\begin{thm}}{\end{thm}}
\crefname{thm}{theorem}{theorems}
\newenvironment{theoremC}{\begin{thmC}}{\end{thmC}} 
\crefname{thmC}{theorem}{theorems}
\newenvironment{corollary}{\begin{cor}}{\end{cor}}
\crefname{cor}{corollary}{corollaries}
\newenvironment{lemma}{\begin{lem}}{\end{lem}}
\crefname{lem}{lemma}{lemmas}

\crefname{prop}{proposition}{propositions}
\newenvironment{remark}{\begin{rem}}{\end{rem}}
\crefname{rem}{remark}{remarks}
\newenvironment{example}{\begin{exa}}{\end{exa}}
\crefname{exa}{example}{examples}
\newenvironment{definition}{\begin{defi}}{\end{defi}}
\crefname{defi}{definition}{definitions}

\begin{document}

\title[Model Checking Recursive Probabilistic Programs]{Model Checking Temporal Properties of\texorpdfstring{\\}{} Recursive Probabilistic Programs}
\titlecomment{{\lsuper*}This is an extended version of a conference paper with the same title published at FoSSaCS 2022~\cite{fossacs}.}
\thanks{This work is supported by the DFG Research Training Group 2236 UnRAVeL and the ERC Advanced Grant 787914 FRAPPANT}

\author[T.~Winkler]{Tobias Winkler\lmcsorcid{0000-0003-1084-6408}}
\author[C.~Gehnen]{Christina Gehnen\lmcsorcid{0000-0002-6548-3432}}
\author[J.-P.~Katoen]{Joost-Pieter Katoen\lmcsorcid{0000-0002-6143-1926}}

\address{RWTH Aachen University, Aachen, Germany}	
\email{\{tobias.winkler,christina.gehnen,katoen\}@cs.rwth-aachen.de}  

\begin{abstract}
    Probabilistic pushdown automata (pPDA) are a standard operational model for programming languages involving discrete random choices and recursive procedures.
    Temporal properties are useful for specifying the chronological order of events during program execution.
    Existing approaches for model checking pPDA against temporal properties have focused mostly on $\omega$-regular and LTL properties.
    In this paper, we give decidability and complexity results for the model checking problem of pPDA against $\omega$-visibly pushdown languages that can be described by specification logics such as CaRet.
    These logical formulae allow specifying properties that explicitly take the structured computations arising from procedural programs into account.
    For example, CaRet is able to match procedure calls with their corresponding future returns, and thus allows to express fundamental program properties such as total and partial correctness.
\end{abstract}


\maketitle


\section{Introduction}
\label{sec:intro}

\emph{Probabilistic programs} extend traditional programs with the ability to flip coins or, more generally, sample values from probability distributions.
These programs can be used to encode randomized algorithms and mechanisms in security~\cite{DBLP:journals/toplas/BartheKOB13} in a natural way.
The interest in probabilistic programs has significantly increased in recent years.
To a large extent, this is due to the search in AI for more expressive and succinct languages than probabilistic graphical models for Bayesian inference~\cite{DBLP:conf/icse/GordonHNR14}.
Probabilistic programs have many applications~\cite{DBLP:journals/corr/abs-1809-10756}.
They are used in, among other areas, machine learning, systems biology, security, planning and control, quantum computing, and software--defined networks.
Probabilistic variants of many programming languages exist.

\paragraph{Recursion}
\emph{Procedural programs} allow for the declaration of procedures---small independent code blocks---and the ability to \emph{call} procedures from one another, possibly in a recursive fashion.
Most common programming languages such as C, Python, or Java support procedures.
It is thus not surprising that recursion is also present in many modern probabilistic programming languages (PPL) such as WebPPL~\cite{dippl} or Church~\cite{DBLP:conf/starai/StuhlmullerG12}.
In fact, there have been numerous approaches to extend Bayesian networks with recursion even before PPL became popular~\cite{DBLP:conf/aaai/PfefferK00,DBLP:journals/amai/Jaeger01,Cassini2008}.
Randomized algorithms such as Hoare's quicksort (see, e.g.,~\cite{DBLP:journals/dam/Karp91}) with random pivot selection can be readily implemented using recursion.
Finally, recursion in form of branching processes is an important tool to model reproduction of cells or molecules in systems biology~\cite{axelrod2015branching}.

\begin{figure}[t]
    \centering
    \begin{minipage}{1\linewidth}
        \begin{minipage}[t]{0.45\linewidth}
            \begin{align*}
            &\texttt{proc} ~ \texttt{void} ~ \mathit{infectYoung}(): \\
            &\qquad y \coloneqq \texttt{uniform}(0,3) \\
            &\qquad \texttt{repeat} ~ y ~ \texttt{times}: \\
            &\qquad \qquad \mathit{infectYoung}() \\
            &\qquad e \coloneqq \texttt{uniform}(0,2) \\
            &\qquad \texttt{repeat} ~ e ~ \texttt{times}: \\
            &\qquad \qquad f \coloneqq \mathit{infectElder}() \\
            &\qquad \texttt{return}
            \end{align*}
        \end{minipage}
        \hspace{1cm}%
        \begin{minipage}[t]{0.45\linewidth}
            \begin{align*}
            &\texttt{proc} ~ \texttt{bool} ~ \mathit{infectElder}(): \\
            &\qquad y \coloneqq \texttt{uniform}(0,1) \\
            &\qquad \texttt{repeat} ~ y ~ \texttt{times}: \\
            &\qquad \qquad \mathit{infectYoung}()  \\
            &\qquad e \coloneqq \texttt{uniform}(0,4) \\
            &\qquad \texttt{repeat} ~ e ~ \texttt{times}: \\
            &\qquad \qquad \mathit{infectElder}() \\
            &\qquad f \coloneqq \texttt{bernoulli}(0.01) \\ 
            &\qquad \texttt{return} ~ f
            \end{align*}
        \end{minipage}
    \end{minipage}
    \caption{
        Recursive probabilistic program modeling the outbreak of an infectious disease.
        $\texttt{uniform}(a,b)$ stands for the \emph{discrete} uniform distribution on $[a,b]$.
    }
    \label{fig:exampleProg}
\end{figure}

\begin{wrapfigure}[8]{r}{0.37\textwidth}
    \centering
    \setlength{\tabcolsep}{1em}
    \renewcommand{\arraystretch}{1.2}
    \begin{tabular}{c | c c}
         & Y & E \\
        \hline
        Y & 1.5 & 1 \\
        E & 0.5 & 2
    \end{tabular}
    \caption{Example infection rates by age groups.}
    \label{fig:tableRates}
\end{wrapfigure}

\paragraph{Motivating example}
This paper studies the automated verification of \emph{probabilistic pushdown automata} (pPDA)~\cite{esparzaMCPPDA-lics} as an explicit-state operational model of procedural probabilistic programs against temporal specifications.
As a motivating example we consider a simple epidemiological model for the outbreak of an infectious disease in a large population where the number of susceptible individuals can be assumed to be infinite.
Our example model distinguishes young and elderly persons.
Each affected individual infects a uniformly distributed number of others, with varying rates (expected values) according to the age groups (\Cref{fig:tableRates}).
The fatality rate for infected elderly and young persons is 1\% and 0\%, respectively.
Procedure $\mathit{infectElder()}$ returns a Boolean in order to signal to its callers whether the infection has led to fatality.
Initially, we assume there is a single infected young person, i.e., the overall program is started by calling $\mathit{infectYoung()}$.
It is an easy exercise to specify this model as a discrete probabilistic program with mutually recursive procedures (\Cref{fig:exampleProg}).
Note that this program can be easily amended to more realistic scenarios involving, e.g., more age or gender groups, hospitalization rate, etc.

The behavior of such recursive probabilistic programs can be naturally described by pPDA.
Roughly, the \emph{local} states of the procedures---the values of the variables in the procedure's scope and the position of the program counter---constitute both the state space and the stack alphabet of the automaton.
Procedure calls correspond to push transitions in the pPDA in such a way that the program's \emph{procedure stack} is simulated by the automaton's pushdown stack, i.e., the caller's current local state is saved on top of the stack.
Accordingly, \emph{returning} from a procedure corresponds to taking a pop transition in order to restore the local state of the caller.
Returning a value can be handled similarly.
Clearly, if the reachable local state spaces of the involved procedures are finite, then the resulting automaton will be finite as well.
We refer to~\cite{DBLP:reference/mc/AlurBE18} for more details.

A number of natural questions such as ``Will the virus eventually become extinct?'' (termination probability) or ``What is the expected number of fatalities?'' (expected costs) are decidable for finite pPDA (see~\cite{brazdilSurvey} for a survey).
In this work, we focus on \emph{temporal properties}, i.e., questions that involve reasoning about the chronological order of events during the epidemic.
An example are chains of infection:
For instance, we might ask
\begin{highlight}%
    What's the probability that eventually a young person with only young persons in their chain of infection passes the virus on to an elderly person who then dies?
\end{highlight}
On the level of the program in \Cref{fig:exampleProg}, this corresponds to the probability of reaching a \emph{global} program configuration where the call stack only contains $\mathit{infectYoung}()$ invocations and during execution of the current $\mathit{infectYoung}()$, the local variable $f$ is eventually set to true.
This requires reasoning about the nestings of calls and returns of a computation.
In fact, in order to decide if $f = \mathit{true}$ in the current procedure, we must ``skip'' over all calls within it and only consider their local return values.
This requirement (and many others) can be naturally expressed in the logic \emph{CaRet}~\cite{caret}, an extension of LTL:
\[
    \finally^g\, (\, \globally^- p_Y ~\land~ p_Y ~\land~ \finally^a f \,) ~.
\]
Here, $p_Y$ is an atomic proposition that holds at states which correspond to being in procedure $\mathit{infectYoung}$, and $f$ indicates that $f = \mathit{true}$.
Intuitively, the above formula states that eventually (outer $\finally^g$), the computation reaches a (global) state where only $\mathit{infectYoung}$ is on the call stack and the current procedure is $\mathit{infectYoung}$ as well ($\globally^- p_Y ~\land~ p_Y$), and moreover the local---aka \emph{abstract}---path \emph{within} in the current procedure reaches a state where $f$ is true ($\finally^a f$).
Such properties are in general context-free but not always regular and thus cannot be expressed in LTL~\cite{caret}.

\paragraph{Contributions}
The contribution of this paper is a solution to the following problem:
\begin{highlight}%
    Given a (finite) pPDA $\ppda$ and an $\omega$-visibly pushdown language (VPL) $\mathcal L$ in terms of either a CaRet formula or an automaton, determine the probability that a random trajectory of $\ppda$ is in $\mathcal L$.
\end{highlight}
The complexity results for the associated decision problems are summarized in~\Cref{fig:resTable}.
As common in the literature, we consider the special case of qualitative, i.e., \emph{almost-sure} model checking separately.
To the best of our knowledge, none of the problems in \Cref{fig:resTable} was known to be decidable before.
The work of~\cite{dubslaff} proved decidability of model checking against deterministic Muller \emph{visibly pushdown automata} (VPA) which capture a strict subset of the CaRet-definable languages~\cite{vpl}.
The most important technical insight of this paper is that two existing (but independently developed) constructions from the literature can be combined to enable effective model checking of pPDA against $\omega$-VPL:
The \emph{deterministic stair-parity VPA} introduced in \cite{vpGames}, and a certain \emph{finite Markov chain} associated with a pPDA~\cite{esparzaMCPPDA-lics}.
We provide some more details in the next paragraph.

\begin{table}[t]
    \caption{
        Complexity results established in this paper.
    }
    \label{fig:resTable}
    \centering
    \setlength{\tabcolsep}{3pt} 
    \begin{tabular}{l l l}
        \toprule
        \emph{$\omega$-VPL given in terms of ...} & \emph{qualitative} & \emph{quantitative} \\
        \midrule
        Deterministic stair-parity VPA~[\Cref{thm:resStPaDVPA}] \quad \quad & in $\PSPACE$ &  in $\PSPACE$ \\[5pt]
        Non-deterministic Büchi VPA~[\Cref{thm:resNVPAQual}] & $\EXPTIME$-compl. \quad\quad& in $\EXPSPACE$ \\[5pt]
        CaRet formula~[\Cref{thm:resCaRet}] & in $\TwoEXPTIME$ & in $\TwoEXPSPACE$ \\
        \bottomrule
    \end{tabular}
\end{table}

\paragraph{Techniques and tools}
We briefly outline our approach which is built on a number of existing constructions and results from the literature.
In order for the model checking problems to be decidable~\cite{dubslaff}, we need to impose a mild \emph{visibility} restriction on $\ppda$, yielding a probabilistic \emph{visibly} pushdown automaton (pVPA).
Just like several previous works on model checking pPDA against $\omega$-regular specifications~\cite{esparzaMCPPDA-lics,BKS05,esparzaMCPPDA}, we follow an automata-based approach (see~\Cref{fig:reductionsOverview}).
More specifically, we first translate $\varphi$ into an equivalent non-deterministic \emph{Büchi} VPA~\cite{vpl} $\npda$ and then determinize it using a procedure introduced by Löding et al.~\cite{vpGames}.
The resulting DVPA $\dvpa$ uses a so-called \emph{stair-parity}~\cite{vpGames} acceptance condition that is strictly more expressive than standard parity or Muller DVPA~\cite{vpl}.
Stair-parity differs from usual parity in that it only considers certain positions---called \emph{steps}~\cite{vpGames}---of an infinite word where the stack height never decreases again.
We then construct a standard product $\automProd{\ppda}{\dvpa}$.
Here, the visibility conditions ensure that the automata synchronize their stack actions, yielding a product automaton that uses \emph{a single stack} instead of two independent stacks, which would lead to undecidability~\cite{dubslaff}.
Finally, we are left with computing a stair-parity acceptance probability in the product.
This is achieved by constructing a specific finite Markov chain associated to $\automProd{\ppda}{\dvpa}$, called \emph{step chain} in this paper.
Intuitively, the step chain ``jumps'' from one step of a run to the next, hence we only need to evaluate \emph{standard parity} on the step chain.
The idea of step chains was introduced by Esparza et al.~\cite{esparzaMCPPDA-lics} who used them to show decidability of the model checking problem against deterministic (non-pushdown) Büchi automata. 
For constructing the step chain, certain \emph{reachability probabilities} in the given pPDA need to be computed.
These probabilities are algebraic numbers (i.e., solutions of polynomial equations) that may be irrational in general.
However, the relevant problems are still decidable via an encoding in the existential fragment of the FO-theory of the reals (ETR)~\cite{esparzaMCPPDA}.

\begin{figure}[t]
    \begin{tikzpicture}[node distance=-3mm and 5mm,shorten >=1pt]
    \tikzstyle{mybox} = [rectangle,draw=black,align=center,rounded corners,inner sep=4pt]
    \tikzstyle{mytrans} = [->]
    \node[mybox] (caret) {CaRet\\formula $\varphi$\\(Def.~\labelcref{def:caret})};
    \node[mybox,right=of caret] (nvpa) {Büchi\\NVPA $\npda$\\(Def.~\labelcref{def:vpa})};
    \node[mybox,right=of nvpa] (dvpa) {Stair-parity\\DVPA $\dpda$\\(Def.~\labelcref{def:stairParity})};
    \node[mybox,above right=of dvpa] (prod) {Product\\$\automProd{\ppda}{\dpda}$\\(Def.~\labelcref{def:prod})};
    \node[mybox,above left=of prod] (ppda) {pVPA $\ppda$\\(Def.~\labelcref{def:pVPA})};
    \node[mybox,left=of ppda,dashed,minimum height=11mm] (program) {Program};
    \node[mybox,right=of prod] (stepchain) {Step chain\\$\stepMC{\automProd{\ppda}{\dpda}}$\\(Def.~\labelcref{def:stepMC})};
    \node[mybox,right=of stepchain] (etr) {ETR};
    
    \draw[mytrans] (caret) -- (nvpa);
    \draw[mytrans] (nvpa) -- (dvpa);
    \draw[mytrans] (dvpa) -- (prod);
    \draw[mytrans] (ppda) -- (prod);
    \draw[mytrans] (program) edge[dashed] (ppda);
    \draw[mytrans] (prod) -- (stepchain);
    \draw[mytrans] (stepchain) -- (etr);
    \end{tikzpicture}
    \caption{Chain of reductions used in this paper.
        ETR stands for \emph{existential theory of the reals}, i.e., the existentially quantified fragment of the FO-theory over $(\mathbb{R}, +, \cdot, \leq)$.}
    \label{fig:reductionsOverview}
\end{figure}

\paragraph{Related work.}
We have already mentioned various works on recursion in probabilistic graphical models and PPL as well as on verifying pPDA and the equivalent model of recursive Markov chains~\cite{DBLP:journals/jacm/EtessamiY09}.
The analysis of these models focuses on reachability probabilities, $\omega$-regular properties or (fragments of) probabilistic CTL, expected costs, and termination probabilities. 
The computation of termination probabilities in recursive Markov chains and variations thereof with non-determinism is supported by the software tool PReMo~\cite{DBLP:conf/tacas/WojtczakE07}. 
Our paper can be seen as a natural extension from checking pPDA against $\omega$-regular properties to $\omega$-visibly pushdown languages.
In contrast to these algorithmic approaches, various deductive reasoning methods have been developed for recursive probabilistic programs. 
Proof rules for recursion were first provided in~\cite{DBLP:phd/ethos/Jones90}, and later extended to proof rules in a weakest-precondition reasoning style~\cite{DBLP:journals/tcs/McIverM01a,DBLP:conf/lics/OlmedoKKM16}.  
The authors of~\cite{DBLP:conf/lics/OlmedoKKM16} also address the connection to pPDA and provide proof rules for expected run-time analysis.
A mechanized method for proving properties of randomized algorithms, including recursive ones, for the Coq proof assistant is presented in~\cite{DBLP:journals/scp/AudebaudP09}.
The Coq approach is based on higher-order logic using a monadic interpretation of programs as probabilistic distributions.

\paragraph{Conference version}
A preliminary version of this article was published at FoSSaCS 2022~\cite{fossacs}.
The present journal version extends the conference paper by the full proofs as well as further examples and explanations.

\paragraph{Paper structure}
This paper is structured as follows.
We review the basics about VPA and CaRet in \Cref{sec:prelims}.
\Cref{sec:pvpa} introduces probabilistic \emph{visibly} pushdown automata (pVPA).
The stair-parity DVPA model checking procedure is presented in \Cref{sec:modelCheckingDVPA}, and the results for Büchi VPA and CaRet in \Cref{sec:modelCheckingNVPACaRet}.
We conclude the paper in \Cref{sec:conclusion}.


\section{Visibly Pushdown  Languages}
\label{sec:prelims}

In this section, we summarize some preliminary results on visibly pushdown languages and their corresponding automata models, and we recall the syntax and semantics of CaRet.

We use the following notation for words.
Given a non-empty alphabet $\ab$, let $\ab^*$ be the set of all finite words (including the empty word $\epsilon$), and let $\ab^\omega$ be the set of all infinite words over $\ab$.
For $i \geq 0$, the $i$-th symbol of a word $\word \in \ab^* \cup \ab^\omega$ is denoted $w(i)$ if it exists.
$|\word|$ denotes the length of $\word$.
For $n\in\Nats$ we write $\ab^n = \{\word \in \ab^* \mid |\word| = n\}$.
For sets of words $A \subseteq \ab^*$ and $B \subseteq \ab^* \cup \ab^\omega$, the concatenation of all words from $A$ with those from $B$ is denoted $\langConcat{A}{B}$.
We also use $\langConcat{a}{B}$ and $\langConcat{A}{b}$ as shorthands for $\langConcat{\set{a}}{B}$ and $\langConcat{A}{\set{b}}$, respectively.

\subsection{Visibly Pushdown Automata}
\label{sec:vpa}

A finite alphabet $\ab$ is called \emph{pushdown alphabet} if it is equipped with a partition $\ab = \partition{\abCall}{\abInt}{\abRet}$ into three---possibly empty---subsets of \emph{call}, \emph{internal}, and \emph{return} symbols.
A \emph{visibly pushdown automaton} (VPA) over $\ab$ is like a standard pushdown automaton with the additional syntactic restriction that reading a call, internal, or return symbol triggers a push, internal, or pop transition, respectively (an internal transition is one that does not change the stack height).
Formally:

\begin{definition}[{VPA~\cite{vpl}}]
    \label{def:vpa}
    Let $\ab$ be a pushdown alphabet.
    A \emph{visibly pushdown automaton (VPA)} over $\ab$ is a tuple
    $\npda = \npdaInit$
    with $\npdaStates$ a finite set of states,
    $\sInit \in \npdaStates$ an initial state,
    $\abStack$ a finite stack alphabet,
    $\botStack \in \abStack$ a special bottom-of-stack symbol, and
    $\npdaTrans = (\npdaTrans_{\callSubscript}, \npdaTrans_{\intSubscript}, \npdaTrans_{\retSubscript})$ a triple of relations such that
    \[
        \npdaTrans_{\callSubscript} \subseteq  (\npdaStates \times \abCall) \times (\npdaStates \times \abStackNoBot) ~,
        \quad
        \npdaTrans_{\intSubscript} \subseteq  (\npdaStates \times \abInt) \times \npdaStates ~,
        \quad
        \npdaTrans_{\retSubscript} \subseteq  (\npdaStates \times \abRet \times \abStack) \times \npdaStates
    \]
    where $\abStackNoBot = \abStack \setminus \set{\botStack}$.
    \qeddef
\end{definition}

For $s, t \in \npdaStates$, $Z \in \abStack$, and $a \in \ab$, we write
$\transCall{s}{a}{t}{Z}$, $\transInt{s}{a}{t}$, $\transRet{s}{Z}{a}{t}$ to indicate that there exist transitions
$(s,a,t,Z) \in \npdaTrans_{\callSubscript}$, $(s,a,t) \in \npdaTrans_{\intSubscript}$, $(s,a,Z,t) \in \npdaTrans_{\retSubscript}$, respectively.

The semantics of a VPA is defined as usual via configurations and runs.
A \emph{configuration} of VPA $\npda$ is a tuple $(s,\stack) \in \npdaStates \times \langConcat{\abStackNoBot^*}{\botStack}$, written more succinctly as $s\stack$ in the sequel.
Intuitively, being in configuration $s\stack$ means that the automaton is currently in state $s$ and has the word $\stack$ on its stack.
If $s \stack = s Z \stack'$ for some $Z \in \abStack$, then $sZ$ is called the \emph{head} of $s \stack$.
A \emph{bottom configuration} is a configuration with head $s \botStack$ for some $s \in \npdaStates$.
Let $\word \in \ab^\omega$ be an infinite input word.
An infinite sequence $\run = s_0\stack_0,s_1\stack_1 \ldots$ of configurations is called a \emph{run} of $\npda$ on $\word$ if $s_0\stack_0 = \sInitNpda\botStack$ and for all $i\geq 0$, exactly one of the following cases applies:
\begin{itemize}
    \item $\word(i) \in \abCall$ and $\stack_{i+1} = Z \stack_i$ for some $Z \in \abStackNoBot$ such that $\transCall{s_i}{\word(i)}{s_{i+1}}{Z}$; or
    \item $\word(i) \in \abInt$ and $\stack_{i+1} = \stack_i$ and $\transInt{s_i}{\word(i)}{s_{i+1}}$; or
    \item $\word(i) \in \abRet$ and $Z \stack_{i+1} = \stack_i$ for some $Z \in \abStackNoBot$ such that $\transRet{s_i}{Z}{\word(i)}{s_{i+1}}$; or
    \item $\word(i) \in \abRet$ and  $\gamma_i=\gamma_{i+1} = \botStack$ and $\transRet{s_i}{\botStack}{\word(i)}{s_{i+1}}$.
\end{itemize}
In other words, if $\npda$ reads a call (or internal) symbol $a$ while being in configuration $s\stack$ and there exists a suitable call transition $\transCall{s}{a}{t}{Z}$ (or internal transition $\transInt{s}{a}{t}$), then a run of $\npda$ may evolve from configuration $s \stack$ to $t Z\stack$ (or $t \stack$, respectively).
Similarly, if $\npda$ reads a return symbol $a$ in configuration $s Z \stack$ where $Z \neq \botStack$ and there is a transition $\transRet{s}{Z}{a}{t}$, then a run can move from $sZ \stack$ to $t \stack$.
Note that invoking a return transition in a bottom configuration $s\botStack$ does not remove the topmost symbol $\botStack$ from the stack.

A \emph{Büchi acceptance condition} for $\npda$ is a subset $\buchiset \subseteq \npdaStates$.
A VPA equipped with a Büchi condition is called a \emph{Büchi VPA}.
An infinite word $\word \in \ab^\omega$ is accepted by a Büchi VPA if there exists a run $s_0\stack_0,s_1\stack_1,\ldots$ of $\npda$ on $\word$ such that $s_i \in F$ for infinitely many $i \geq 0$.
The $\omega$-language of words accepted by a Büchi VPA $\npda$ is denoted $\lang{\npda} \subseteq \ab^\omega$.

\begin{definition}[{$\omega$-VPL~\cite{vpl}}]
    Let $\ab$ be a pushdown alphabet.
    $L \subseteq \ab^\omega$ is an \linebreak \emph{$\omega$-visibly pushdown language ($\omega$-VPL)} if $L = \lang{\npda}$ for a Büchi VPA $\npda$ over $\ab$.
    \qeddef
\end{definition}

A VPA is \emph{deterministic} (DVPA) if the relations $\npdaTrans_{\callSubscript}, \npdaTrans_{\intSubscript}$, and $\npdaTrans_{\retSubscript}$ are total functions, i.e.,
$\npdaTrans_{\callSubscript} \colon  (\npdaStates \times \abCall) \to (\npdaStates \times \abStackNoBot),
\npdaTrans_{\intSubscript} \colon  (\npdaStates \times \abInt) \to \npdaStates$,
and
$
\npdaTrans_{\retSubscript} \colon  (\npdaStates \times \abRet \times \abStack) \to \npdaStates$.
Note that DVPA have exactly one run on each input word.
As for standard NBA, the class of languages recognized by Büchi DVPA is a strict subset of the languages recognized by non-deterministic Büchi VPA.
Unlike in the non-pushdown case, DVPA with Muller or parity conditions are also strictly less expressive than non-deterministic Büchi VPA~\cite{vpl}.
A deterministic automaton model for $\omega$-VPL was given in~\cite{vpGames}.
It uses a so-called \emph{stair-parity} acceptance condition which we explain in the upcoming~\Cref{sec:stair-parity}.

\subsection{Steps and Stair-parity Conditions}
\label{sec:stair-parity}

In the remainder of this section, $\ab$ denotes a pushdown alphabet and $\npda$ a VPA over $\ab$.
Consider a run $\run = s_0\stack_0,s_1\stack_1,\ldots$ of $\npda$ on an infinite word $\word \in \ab^\omega$.
We define the \emph{stack height} of the $i$-th configuration as $\sh{\run(i)} = |\stack_i| - 1$ (i.e., the bottom symbol $\botStack$ does not count to the stack height).
The stair-parity condition relies on the notion of \emph{steps}:
\begin{definition}[Step]
    \label{def:step}
    Let $\run$ be a run of $\npda$.
    Position $i \geq 0$ is a \emph{step} of $\run$ if 
    \[
        \forall n \geq i  \colon\quad  \sh{\run(n)} ~\geq~ \sh{\run(i)}~. \tag*{\qeddef}
    \]
\end{definition}
Intuitively, a step is a position of a run such that there is no future position where the stack height is strictly smaller.
Slightly abusing terminology, we also say that a \emph{configuration} $s_i\stack_i$ of a given run $\run = s_0\stack_0,s_1\stack_1,\ldots$ is a step if position $i$ is a step.

\begin{example}
    \Cref{fig:stepsExample} depicts a DVPA and the initial fragment of its unique run $\run$ on the input word $\tau \, r \, c \, \tau \, \tau \, c \, r \, c^2 \, r^2 \, c^3 \, r^3 \ldots$.
    The step positions are underlined, i.e., positions 0-5, 7, 11, and 17 are steps.
    Note that if $\run(i) = s \botStack$ for some $s \in \npdaStates$ then $i$ is a step, i.e., bottom configurations are always steps.
\end{example}

\begin{figure}[t]
    \centering
    \begin{minipage}{0.36\textwidth}
        \begin{tikzpicture}[myautomaton]
            \node[state,initial,initial text=,initial where=above] (s0) {$s_0$};
            \node[state,right=of s0] (s1) {$s_1$};
            \draw[->] (s0) edge[loop left] node[above,yshift=1mm] {$c,Z$} (s0);
            \draw[->] (s1) edge[loop right] node[align=center,above,yshift=1mm] {\transStackThree{$\botStack,r$}{$Z,r$}{$\tau$}} (s1);
            \draw[->] (s0) edge[bend left] node[align=center] {\transStackThree{$\botStack,r$}{$Z,r$}{$\tau$}} (s1);
            \draw[->] (s1) edge[bend left] node {$c,Z$} (s0);
        \end{tikzpicture}
    \end{minipage}
    \hfill
    \begin{minipage}{0.62\textwidth}
        \setlength{\tabcolsep}{2pt}
        \begin{tabular}{c c c c c c c c c c c c c c c c c c c}
            $\tau$ & $r$ & $c$ & $\tau$ & $\tau$ & $c$ & $r$ & $c$ & $c$ & $r$ & $r$ & $c$ & $c$ & $c$ & $r$ & $r$ & $r$ & ... & \\[-2pt]
            \tiny 0 & \tiny 1 & \tiny 2 & \tiny 3 & \tiny 4 & \tiny 5 & \tiny 6 & \tiny 7 & \tiny 8 & \tiny 9 & \tiny 10 & \tiny 11 & \tiny 12 & \tiny 13 & \tiny 14 & \tiny 15 & \tiny 16 &  \tiny 17 & ...\\[2pt]
            & & &     &     &     &     &     &     &     &     &     &     &     & $Z$ &     &     &    &\\[-2pt]
            & & &     &     &     &     &     &     & $Z$ &     &     &     & $Z$ & $Z$ & $Z$ &     &    &\\[-2pt]
            & & &     &     &     & $Z$ &     & $Z$ & $Z$ & $Z$ &     & $Z$ & $Z$ & $Z$ & $Z$ & $Z$ &    & ... \\[-2pt]
            & & & $Z$ & $Z$ & $Z$ & $Z$ & $Z$ & $Z$ & $Z$ & $Z$ & $Z$ & $Z$ & $Z$ & $Z$ & $Z$ & $Z$ & $Z$ & ...\\[-2pt]        
            $\botStack$ & $\botStack$ & $\botStack$ & $\botStack$ & $\botStack$ & $\botStack$ & $\botStack$ & $\botStack$ & $\botStack$ & $\botStack$ & $\botStack$ & $\botStack$ & $\botStack$ & $\botStack$ & $\botStack$ & $\botStack$ & $\botStack$ & $\botStack$ & ...\\[3pt] 
            $\underline{s_0}$ & $\underline{s_1}$ & $\underline{s_1}$ & $\underline{s_0}$ & $\underline{s_1}$ & $\underline{s_1}$ & $s_0$ & $\underline{s_1}$ & $s_0$ & $s_0$ & $s_1$ & $\underline{s_1}$ & $s_0$ & $s_0$ & $s_0$ & $s_1$ & $s_1$ & $\underline{s_1}$ & ...
        \end{tabular}
    \end{minipage}
    
    \caption{
        Left: An example VPA (in fact, a DVPA) with $\abStack = \set{Z, \botStack}$ over input alphabet $\ab = \partition{\set{c}}{\set{\tau}}{\set{r}}$.
        Transitions labeled $c,Z$ are \emph{call transitions} that push $Z$ on the stack.
        The transitions labeled $\tau$ are \emph{internal transitions}; they ignore the stack completely.
        Transitions labeled $Z,r$ and $\bot, r$ are \emph{return transitions} that are only enabled if $Z$ ($\bot$, respectively) is on top of the stack.
        When executing $Z,r$, the symbol $Z$ is popped from the stack.
        However, the special bottom-of-stack symbol $\bot$ can never be popped (e.g., position 1 in the run).        
        Right: The \emph{unique} run of the DVPA on input word $\tau \, r \, c \, \tau \, \tau \, c \, r \, c^2 \, r^2 \, c^3 \, r^3 \ldots$.
        Steps are underlined.
    }
    \label{fig:stepsExample}
\end{figure}

Steps play a central role in the rest of the paper.
We therefore explain some of their fundamental properties.
Suppose that $\rho$ is a run of $\npda$ on an infinite word $\word \in \ab^*$.
\begin{itemize}
    \item 
    If positions $i < j$ are \emph{adjacent} steps, i.e., there is no step $k$ strictly in between $i$ and $j$, then $\sh{\run(j)} - \sh{\run(i)} \in \set{0,1}$, i.e., the stack height from one step to the next increases by either zero or one.
    \item Each step $i$ has a \emph{next} step $j > i$:
    If the symbol at step $i$ is internal (e.g., $i = 0,3,4$ in \Cref{fig:stepsExample}) or a return (e.g., $i=1$) then the next step is simply the next position $j = i+1$ and the stack height does not increase.
    If the symbol at position $i$ is a call, then one of two cases occurs:
    Either the call has no matching future return (e.g., $i=2$); 
    in this case, the next step is the next position $j = i+1$.
    Otherwise the call is eventually matched (e.g., $i =5, 7, 11$) and the next step $j > i+1$ occurs after the corresponding matching return is read and has the same stack height.
    \item As a consequence, each infinite run has infinitely many steps.
    Notice though that the difference between two adjacent step positions may grow unboundedly as in \Cref{fig:stepsExample}.
    \item The stack height at the steps either grows unboundedly or eventually stabilizes (the latter occurs in \Cref{fig:stepsExample}; the stack heights at the steps induce the sequence $0,0,0,1^\omega$).
\end{itemize}

\begin{remark}
    \label{rem:shOfWords}
    We can also define the steps of a \emph{word} $\word \in \ab^\omega$ as the positions where a run of \emph{any arbitrary} VPA on $\word$ has a step.
    Due to the visibility restriction, the actual behavior of the VPA does not influence the step positions~\cite{vpGames}.
    In other words, the step positions are predetermined by the input word.
    Thus, we can also speak of the stack height $\sh{w(i)}$ of word $w$ at position $i$.
\end{remark}

We need one last notion before defining stair-parity.
The \emph{footprint} of an infinite run $\run = s_0\stack_0,s_1\stack_1,\ldots$ is the infinite sequence $\stepRestr{\run} = s_{n_0} s_{n_1} \ldots \in S^\omega$ where for all $i \geq 0$ the position $n_i$ is the $i$-th step of $\run$.
In words, $\stepRestr{\run}$ is the projection of the run $\run$ onto the \emph{states} occurring at its steps.
For the example run in \Cref{fig:stepsExample}, $\stepRestr{\run} = s_0s_1s_1s_0s_1^\omega$.

\begin{definition}[Stair-parity~\cite{vpGames}]
    \label{def:stairParity}
    Let $\npda$ be a VPA over pushdown alphabet $\ab$.
    A \emph{stair-parity} acceptance condition for $\npda$ is defined in terms of a priority function $\prioFun \colon \npdaStates \to \Nats$.
    A word $\word \in \ab^\omega$ is accepted if $\npda$ has a run $\run$ on $\word$ such that
    \[
        \min\, \Set{ \prioval }{ \existsInfty i \geq 0 \colon~ \prioFun(\, \stepRestr{\run}(i)\,) \,=\, \prioval } \quad \text{is even.}
    \]
   The language accepted by $\npda$ is denoted $\lang{\npda}$.
   \qeddef
\end{definition}
Intuitively, $\lang{\npda}$ contains all words that have a run $\run$ on $\npda$ such that the minimum priority occurring infinitely often at states in the \emph{footprint} $\stepRestr{\run}$ is even.

\begin{example}
    \label{ex:repbdd}
    The DVPA in \Cref{fig:stepsExample} with $\prioFun(s_0) = 1$ and $\prioFun(s_1) = 2$ accepts 
    \[
        \mathcal{L}_{\mathit{repbdd}} ~=~ \Set{\word \in \ab^\omega}{\exists B \geq 0 ,\, \existsInfty i \geq 0 \colon ~ \sh{w(i)} \,=\, B} ~,
    \]
    the language of \emph{repeatedly bounded} words~\cite{vpGames}, i.e., words whose stack height (cf.\ \Cref{rem:shOfWords}) is infinitely often equal to some constant $B$.
    The example word from~\Cref{fig:stepsExample} satisfies this property with $B = 1$.
    To see why the automaton accepts $\mathcal{L}_{\mathit{repbdd}}$, note that a word is repeatedly bounded iff the stack height at the steps stabilizes eventually.
    The latter occurs iff in just finitely many cases, the transition before reaching a step was a call.
    The DVPA in \Cref{fig:stepsExample} detects this behavior; when reading a call symbol, it always moves to state $s_0$ which has odd priority, and it accepts iff $s_0$ is visited finitely often \emph{at call positions}.
    It is known that $\mathcal{L}_{\mathit{repbdd}}$ is not expressible through DVPA with usual parity conditions~\cite{vpl}.
\end{example}

\begin{theoremC}[{\cite[Theorem 1]{vpGames}}]
    \label{thm:StPaDVPA}
    For every non-deterministic Büchi VPA $\npda$ there exists a deterministic stair-parity DVPA $\dpda$ with $2^{\O(|\npdaStates|^2)}$ states such that $\lang{\npda} = \lang{\dpda}$.
    Moreover, $\dpda$ can be constructed in exponential time in the size of $\npda$.
\end{theoremC}
It was also shown in \cite{vpGames} that stair-parity DVPA characterize exactly the class of $\omega$-VPL (and are thus not more expressive than non-deterministic Büchi VPA).

\subsection{CaRet, a Temporal Logic of Calls and Returns}

Specifying requirements directly in terms of automata is tedious in practice.
\emph{CaRet}~\cite{caret} is an extension of Linear Temporal Logic (LTL)~\cite{DBLP:conf/focs/Pnueli77} that can be used to describe certain $\omega$-VPL.

\begin{definition}[{Syntax of CaRet~\cite{caret}}]
    \label{def:caret}
    Let $\AP$ be a finite set of atomic propositions.
    The logic CaRet adheres to the grammar
    \[
    \varphi \ccoloneqq 
    p \bmid 
    \varphi \lor \varphi \bmid
    \neg \varphi \bmid
    \gnext \varphi \bmid
    \varphi \, \guntil \varphi \bmid
    \anext \varphi \bmid
    \varphi \, \auntil \varphi \bmid
    \caller \varphi \bmid
    \varphi \, \calleruntil \varphi ~,
    \]
    where $p \in \AP \cup \typeset$.
    \qeddef
\end{definition}
Other common modalities such as $\finally^b$ and $\globally^b$ for $b \in \set{g, a, -}$ are defined as usual via $\finally^b \varphi = \mathit{true} \, \until^b \, \varphi$, and $\globally^b \varphi = \neg \finally^b \neg \varphi$.
We now explain the intuitive semantics of CaRet, the formal definition is stated further below in \Cref{def:caretSem}.
We assume familiarity with LTL (see, e.g., \cite[Ch.~5]{MCbible} for an introduction).
CaRet formulae are interpreted over infinite words from the pushdown alphabet $\ab = 2^{\AP} \times \typeset$.
$\gnext$ and $\guntil$ are the standard next and until modalities from LTL (called \emph{global} next and until in CaRet).
CaRet extends LTL by two key operators, the \emph{caller} modality $\caller$ and the \emph{abstract successor} $\anext$.
The semantics of these operators is visually explained in \Cref{fig:caretNexts}.
The caller $\caller$ is a \emph{past} modality that points to the position of the last pending call (if such a call exists).
For internal and return symbols, the abstract successor $\anext$ behaves like $\gnext$ unless the latter is a return, in which case $\anext$ is undefined (e.g., position 3 in \Cref{fig:caretNexts}).
On the other hand, the abstract successor of a call symbol is its \emph{matching return} if it exists, or undefined otherwise.
The caller and abstract successor modalities induce sequences of positions which we call \emph{caller path} and \emph{abstract path}, respectively.
The caller path is always finite and the abstract path can be either finite or infinite.
The until modalities $\calleruntil$ and $\auntil$ are then defined analogously to the standard until $\guntil$ with the difference that they are interpreted over the caller and abstract path, respectively.

A prime application of CaRet is to express \emph{total correctness} of a procedure $F$~\cite{caret}:
\[
    \varphi_{\mathit{total}} \qeq \globally^g \, (\, \callSubscript \,\land\, p \,\land\, p_F ~\rightarrow~ \anext q \,)
\]
where $p$ and $q$ are atomic propositions that hold at the states where the pre- and post-condition is satisfied, respectively, and $p_F$ is an atomic proposition marking the calls to $F$.
Another example is the language of repeatedly bounded words from \Cref{ex:repbdd}; it is described by the formula $\varphi_{\mathit{repbdd}} = \finally^g \globally^g (\callSubscript \rightarrow \anext \retSubscript)$ which states that all but finitely many calls have a matching return.
Further examples are given in~\cite{caret}.

\begin{figure}[t]
    \centering
    \begin{tikzpicture}[myautomaton,node distance=4mm and 13mm,arrowlabel/.style={sloped, above}]
    \node[callState] (0) {$0$};
    \node[state,above right=of 0] (1) {$1$};
    \node[callState, right=of 1] (2) {$2$};
    \node[state,above right=of 2] (3) {$3$};
    \node[retState, right=of 3] (4) {$4$};
    \node[state,below right=of 4] (5) {$5$};
    \node[right= 0.5cm of 5] (...) {$...$};
    
    \draw[->] (0) edge [bend left=45,dotted] node[ fill=white, anchor=center, pos=0.5] {$\bigcirc^g$} (1);
    \draw[->] (0) edge [] node [arrowlabel] {$\callSubscript$} (1);
    
    \draw[->] (1) edge [bend left=50,dotted] node[ fill=white, anchor=center, pos=0.5] {$\bigcirc^{g/a}$} (2);
    \draw[->] (1) edge [] node [arrowlabel] {$\intSubscript$} (2);
    \draw[->] (1) edge [bend left=35,dotted] node[ fill=white, anchor=center, pos=0.5] {$\bigcirc^-$} (0);
    
    \draw[->] (2) edge [bend left=45,dotted] node[ fill=white, anchor=center, pos=0.5] {$\bigcirc^g$} (3);
    \draw[->] (2) edge [] node[arrowlabel] {$\callSubscript$} (3);
    \draw[->] (2) edge [bend left,dotted] node[ fill=white, anchor=center, pos=0.5] {$\bigcirc^-$} (0);
    \draw[->] (2) edge [bend right = 35,dotted] node[ fill=white, anchor=center, pos=0.5] {$\bigcirc^a$} (4);
    
    \draw[->] (3) edge [bend left=45,dotted] node[ fill=white, anchor=center, pos=0.5] {$\bigcirc^g$} (4);
    \draw[->] (3) edge [] node [arrowlabel] {$\intSubscript$} (4);
    \draw[->] (3) edge [bend left=35,dotted] node[ fill=white, anchor=center, pos=0.5] {$\bigcirc^-$} (2);
    
    \draw[->] (4) edge [bend left=60,dotted] node[ fill=white, anchor=center, pos=0.5] {$\bigcirc^{g/a}$} (5);
    \draw[->] (4) edge [] node [arrowlabel] {$\retSubscript$} (5);
    \draw[->] (4) edge [bend left=60,dotted] node[ fill=white, anchor=center, pos=0.5] {$\bigcirc^-$} (2);
    
    \draw[->] (5) edge [in=-20,out=200,dotted] node[ fill=white, anchor=center, pos=0.5] {$\bigcirc^-$} (0);
    \end{tikzpicture}
    \vspace{-1em}
    \caption{
        CaRet's various next modalities applied to the initial fragment of an example word.
        Call, internal, and return positions are depicted as boxes, circles, and rhombs, respectively.
        Note that $\anext$ of position 3 is undefined because $\gnext$ is a return.
        Whether or not $\anext$ of position 0 is defined depends on how the words continues after position 5; more specifically, it is defined iff there occurs a return position on the same height as position 5. In this case, $\anext$ of position 0 will point to the first such occurrence.
    }
    \label{fig:caretNexts}
\end{figure}

We now define the semantics of CaRet formally.
Let $\ab = 2^\AP \times \typeset$ be the pushdown alphabet and $w \in \ab^\omega$.
If $w(i) \in \abCall$, then let $\matchingReturn{\word}{i}$ be the position of the matching return of $w(i)$, or $\undefined$ if there is no matching return.
The \emph{abstract successor} $\asucc{\word}{i}$ of position $i$ in word $\word$ is defined as follows:
\begin{align*}
\asucc{\word}{i}
\eeq
\begin{cases}
    \matchingReturn{\word}{i} & \text{ if } \word(i) \in \abCall \\
    i+1 & \text{ if } \word(i) \notin \abCall \lland w(i{+}1) \notin \abRet\\
    \undefined & \text{ if } \word(i) \notin \abCall \lland w(i{+}1) \in \abRet\\
\end{cases}
\end{align*}
The \emph{caller} $\callersucc{\word}{i}$ of position $i$ in word $\word$ is defined as the innermost (``closest'') unmatched call position $j < i$, or $\undefined$ if there was no previous open call.
\begin{definition}[Semantics of CaRet]
    \label{def:caretSem}
    Let $\ab = 2^\AP \times \typeset$ and $w \in \ab^\omega$.
    For all CaRet formulae $\varphi$ and $i \in \Nats$, we define $\word(i) \models \varphi$ by induction over the structure of $\varphi$ as follows:
    \begin{itemize}
        \item $\word(i) \models p$ iff $\word(i) = (X, \mathsf{type})$ and either $p \in X$ or $p = \mathsf{type}$
        \item $\word(i) \models \varphi_1 \lor \varphi_2$ iff $w(i) \models \varphi_1$ or $w(i) \models \varphi_2$
        \item $\word(i) \models \neg \varphi$ iff $w(i) \nvDash \varphi$
        \item $\word(i) \models \gnext \varphi_1$ iff $w(i+1) \models \varphi_1$
        \item $\word(i) \models \anext \varphi_1$ iff $\asucc{\word}{i} = j \in \Nats$ is defined and $\word(j) \models \varphi_1$
        \item $\word(i) \models \caller \varphi_1$ iff $\callersucc{\word}{i} = j \in \Nats$ is defined and $\word(j) \models \varphi_1$
        \item $\word(i) \models \varphi_1 \, \symboluntil{b} \varphi_2$ for $b \in \set{a,g,-}$ if there exist positions $i = i_0, i_1, \ldots, i_k$, $k \in \Nats$, such that
        \begin{enumerate}
            \item $w(i_k) \models \varphi_2$, and
            \item for all $0 \leq j < k$, $w(i_j) \models \varphi_1$ and $i_{j+1} = \symbolsucc{b}{w}{i_j}$.
        \end{enumerate}
    \end{itemize}
    Furthermore, we write $\word \models \varphi$ if $\word(0) \models \varphi$.
    The language of all words satisfying a CaRet formula $\varphi$ is denoted $\lang{\varphi} = \{\word \in \ab^\omega \mid \word \models \varphi\}$.
    \qeddef
\end{definition}

\begin{theoremC}[{\cite[Theorem 5.1]{DBLP:conf/lics/AlurABEIL07}}]
    \label{thm:caretToNBVPA}
    CaRet-definable languages are $\omega$-VPL:
    For each CaRet formula $\varphi$ there exists a (non-deterministic) Büchi VPA $\npda$ such that $\lang{\varphi} = \lang{\npda}$, and $\npda$ can be constructed in time $2^{\O(|\varphi|)}$.
\end{theoremC}
The above theorem is well-known in the literature~\cite{DBLP:conf/lics/AlurABEIL07,DBLP:reference/mc/AlurBE18} even though it is usually stated for \emph{Nested Word Automata} (NWA) which are equivalent to VPA, and it is more common to state a space bound on $\npda$ rather than a time bound for the construction.
The theorem even applies to the logic NWTL$^+$, an FO-complete extension of CaRet~\cite{DBLP:conf/lics/AlurABEIL07} which we do not consider here for the sake of simplicity.

\Cref{thm:caretToNBVPA,thm:StPaDVPA} together imply that each CaRet formula can be translated to a \emph{deterministic} stair-parity VPA of doubly-exponential size.


\section{Probabilistic Visibly Pushdown Automata}
\label{sec:pvpa}

As explained in the introduction, we employ \emph{probabilistic pushdown automata}~\cite{esparzaMCPPDA-lics} (pPDA) as an operational model for procedural probabilistic programs.
pPDA thus play a fundamentally different role in this paper than VPA (cf.~\Cref{def:vpa}):
pPDA are used to model the \emph{system}, while VPA encode the \emph{specification}.
Consequently, our pPDA do \emph{not} read an input word like VPA do, but instead take their transitions randomly, according to fixed probability distributions.
In this way, they define a probability space over their possible traces, i.e., runs projected on their labeling sequence.
These traces constitute the input words of the VPA.
In order for the model checking problems to be decidable~\cite{dubslaff}, a syntactic visibility restriction needs to be imposed on pPDA.
In a nutshell, the condition is that each state only has outgoing transitions of one \emph{type}, i.e., push, internal, or pop.
This means that the stack operation is \emph{visible in the states} (recall that for VPA, the stack operation is visible in the input symbol).
This restriction is not severe in the context of modeling programs (see~\Cref{rem:pVPA} further below)
and leads to our notion of probabilistic \emph{visibly} pushdown automata (pVPA) which we now define formally.

Given a finite set $X$, we write $\Dist(X) = \Set{ f \colon X \to [0,1] }{ \sum_{a \in X} f(a) = 1 }$ for the set of probability distributions over $X$.

\begin{definition}[pVPA]
    \label{def:pVPA}
    A \emph{probabilistic visibly pushdown automaton} (pVPA) is a tuple $\ppda = \ppdaInit$ where
    $\ppdaStates$ is a finite set of states partitioned into $\ppdaStates = \partition{\ppdaStatesCall}{\ppdaStatesInt}{\ppdaStatesRet}$,
    $\sInitPpda \in \ppdaStates$ is an initial state,
    $\abStack$ is a finite stack alphabet,
    $\botStack \in \abStack$ is a special bottom-of-stack symbol,
    $\ppdaTransFun = (\ppdaTransFun_{\callSubscript}, \ppdaTransFun_{\intSubscript}, \ppdaTransFun_{\retSubscript})$ is a triple of functions
    \[
        \ppdaTransFun_{\callSubscript} \colon \ppdaStatesCall \to \Dist(\ppdaStates \times \abStackNoBot) ~,
        \quad
        \ppdaTransFun_{\intSubscript} \colon \ppdaStatesInt \to \Dist(\ppdaStates) ~,
        \quad
        \ppdaTransFun_{\retSubscript} \colon \ppdaStatesRet \times \abStack \to \Dist( \ppdaStates) ~,
    \]
    $\ab = \partition{\abCall}{\abInt}{\abRet}$ is a pushdown alphabet,
    and $\labeling \colon \ppdaStates \to \ab$ is a state labeling function consistent with the visibility condition, i.e., for all $\visTypeMetaSubscript \in \{\callSubscript, \intSubscript, \retSubscript\}$ and all $q \in \ppdaStates$, it is required that 
    $q \in \ppdaStates_{\visTypeMetaSubscript}$ iff $\labeling(q) \in \ab_{\visTypeMetaSubscript}$.
    \qeddef
\end{definition}

Similar to VPA, we use the notation $\transCall{q}{\prob}{r}{Z}$, $\transInt{q}{\prob}{r}$, and $\transRet{q}{Z}{\prob}{r}$ to indicate that $\ppdaTransFun_{\callSubscript}(q)(r,Z) = \prob$, $\ppdaTransFun_{\intSubscript}(q)(r) = \prob$, and $\ppdaTransFun_{\retSubscript}(q,Z)(r) = \prob$, respectively.

Intuitively, the behavior of a pVPA $\ppda$ is as follows.
If the current state $q$ is a call state, then the probability distribution $\ppdaTransFun_{\callSubscript}(q)$ determines a random successor state and stack symbol to be pushed on the stack ($\botStack$ cannot be pushed).
Similarly, if the current state $q$ is internal, then $\ppdaTransFun_{\intSubscript}(q)$ is the distribution over possible successor states and no stack operation is performed.
Lastly, if the current state $q$ is a return state and symbol $Z \in \abStackNoBot$ is on top of the stack, then $Z$ is popped and $\ppda$ moves to a successor state with probability according to $\ppdaTransFun_{\retSubscript}(q, Z)$.
As in VPA, the special symbol $\botStack$ is \emph{not} popped from the stack if a return transition occurs in a bottom configuration.

The formal semantics of a pVPA is defined in terms of a countably infinite discrete-time \emph{Markov chain}.
A (labeled) Markov chain is essentially the special case of a pVPA where $\ppdaStates = \ppdaStatesInt$, with the only difference that we allow for countably infinite $Q$ and do not impose the restriction on the labeling function $\labeling$.
A Markov chain can thus be specified as a 5-tuple $(\ppdaStates, q_0, P, \ab, \labeling)$, i.e., we omit $\botStack$ and $\abStack$ from the definition of pVPA because a Markov chain does not use a stack.
A \emph{run} of a Markov chain is an infinite sequence of states, i.e., an element from $Q^\omega$.
Note that in our definition, runs do not necessarily start in $q_0$; this is just for technical convenience---impossible runs starting in a state other than $q_0$ will simply have probability $0$.
We extend the labeling function $\labeling$ from states to runs in the natural way.

We define the \emph{Markov chain generated by a pVPA} $\ppda = \ppdaInit$ as
\[
    \ppdaMcSemantic{\ppda} \eeq (\, \ppdaStates \times \langConcat{\abStackNoBot^*}{\botStack},\, q_0\botStack,\, P_{\ppda},\, \ab,\, \labeling_{\ppda} \,) ~,
\]
i.e., the state space of $\ppdaMcSemantic{\ppda}$ is the set of configurations of $\ppda$, and the transition probability function $P_{\ppda}$ is defined as follows.
$P_{\ppda}(q\stack)(r\stack') = p > 0$ iff exactly one of the following cases applies:
\begin{itemize}
    \item $q \in \ppdaStatesCall$ and $\stack' = Z \stack$ for some $Z \in \abStackNoBot$ and $\ppdaTransCall{q}{p}{r}{Z}$; or
    \item $q \in \ppdaStatesInt$ and $\stack' = \stack $ and $\ppdaTransInt{q}{p}{r}$; or
    \item $q \in \ppdaStatesRet$ and $Z \stack' = \stack$ for some $Z \in \abStackNoBot$ and $\ppdaTransRet{q}{Z}{p}{q}$; or
    \item $q \in \ppdaStatesRet$ and
    $\stack' = \stack = \botStack$ and $\ppdaTransRet{q}{\botStack}{p}{r}$.
\end{itemize}
Moreover, the labeling function of $\ppdaMcSemantic{\ppda}$ is $\labeling_{\ppda}(q\stack) = \labeling(q)$ for all $q \stack \in \ppdaStates \times  \langConcat{(\abStackNoBot)^*}{\botStack}$.

\begin{example}
    \Cref{fig:inducedMCexample} depicts a pVPA $\ppda$ and a fragment of its generated Markov chain $\ppdaMcSemantic{\ppda}$.
    Even though $\ppdaMcSemantic{\ppda}$ is infinite, many problems remain decidable, including in particular questions about reachability probabilities which can be characterized as the least solution of a system of polynomial equations~\cite{esparzaMCPPDA-lics}.
    We will use this extensively in \Cref{sec:modelCheckingDVPA}.
\end{example}

\begin{figure}[t]
    \centering
        \begin{tikzpicture}[myautomaton,bend angle=10,node distance=20mm,on grid,initial where=above]
            \node[callState,initial] (q0) {$q_0$};
            \node[retState,below=of q0] (q1) {$q_1$};
            \draw[->] (q0) edge[loop left] node[auto] {$\nicefrac{2}{3}, Z$} (q0);
            \draw[->] (q1) edge[loop left] node[auto,align=center] {$Z, \nicefrac{1}{2}$ \\[-2pt] $\botStack, 1$} (q1);
            \draw[->] (q0) edge[bend right] node[left] {$\nicefrac{1}{3}, Z$} (q1);
            \draw[->] (q1) edge[bend right] node[right] {$Z, \nicefrac{1}{2}$} (q0);
        \end{tikzpicture}
        \hspace{12mm}
        \begin{tikzpicture}[myautomaton,bend angle=10,node distance=20mm,on grid,initial where=above]
            \node[initial] (q0bot) {$q_0\botStack$};
            \node[below=of q0bot] (q1bot) {$q_1\botStack$};
            \node[right=of q0bot] (q0Zbot) {$q_0Z\botStack$};
            \node[right=of q0Zbot] (q0ZZbot) {$q_0ZZ\botStack$};
            \node[right=of q0ZZbot] (q0ZZZbot) {$\cdots$};
            \node[right=of q1bot] (q1Zbot) {$q_1Z\botStack$};
            \node[right=of q1Zbot] (q1ZZbot) {$q_1ZZ\botStack$};
            \node[right=of q1ZZbot] (q1ZZZbot) {$\cdots$};
            
            \draw[->] (q0bot) edge node[auto] {$\nicefrac 2 3$} (q0Zbot);
            \draw[->] (q0bot) edge[bend left] node[auto] {$\nicefrac 1 3$} (q1Zbot);
            \draw[->] (q0Zbot) edge node[auto] {$\nicefrac 2 3$} (q0ZZbot);
            \draw[->] (q0Zbot) edge[bend left] node[auto] {$\nicefrac 1 3$} (q1ZZbot);
            \draw[->] (q0ZZbot) edge node[auto] {$\nicefrac 2 3$} (q0ZZZbot);
            \draw[->] (q0ZZbot) edge[bend left] node[auto] {$\nicefrac 1 3$} (q1ZZZbot);
            
            \draw[->] (q1Zbot) edge node[auto] {$\nicefrac{1}{2}$} (q1bot);
            \draw[->] (q1Zbot) edge[bend left] node[auto] {$\nicefrac{1}{2}$} (q0bot);
            \draw[->] (q1ZZbot) edge node[auto] {$\nicefrac{1}{2}$} (q1Zbot);
            \draw[->] (q1ZZbot) edge[bend left] node[auto] {$\nicefrac{1}{2}$} (q0Zbot);
            
            \draw[->] (q1bot) edge[loop left] node[auto] {$1$} (q1bot);
        \end{tikzpicture}
    \caption{
        Left: A pVPA $\ppda$ (labeling function omitted) with $\abStack = \{\botStack, Z\}$.
        Call and return states are drawn as squares and rhombs, respectively.
        Transitions labeled $p, Z$ for $p \in [0,1]$ mean that with probability $p$ symbol $Z$ is pushed.
        Transitions labeled $X, p$, for $X \in \abStack$ and $p \in [0,1]$, indicate that \emph{if} $X$ is on top of the stack, then the transition is taken with probability $p$ (and $X$ is popped if $X \neq \botStack$).
        Note that the visibility restriction on pVPA enforces that a state may not have both outgoing push and outgoing pop transitions.
        Right: The infinite-state Markov chain $\ppdaMcSemantic{\ppda}$ generated by $\ppda$.
    }
    \label{fig:inducedMCexample}
\end{figure}

We define the set $\Runs_\ppda$ of a pVPA $\ppda$ as the runs of the Markov chain $\ppdaMcSemantic{\ppda}$, i.e., $\Runs_{\ppda} = (\ppdaStates \times \langConcat{\abStackNoBot^*}{\botStack})^\omega$.
\emph{Steps} of pVPA runs are defined as in \Cref{def:step}.
A further example pVPA and its possible runs are depicted in \Cref{fig:step-mc-simple} on page~\pageref{fig:step-mc-simple}.

We associate a probability space with $\Runs_\ppda$ in the usual way (see, e.g.,~\cite[Ch.~10]{MCbible}).
To this end, we define the $\sigma$-algebra $\mathcal{F} \subseteq 2^{\Runs_\ppda}$ as the smallest set that contains all the \emph{cylinder sets} $\langConcat{\run}{\Runs_\ppda}$, where $\run$ is an arbitrary \emph{finite} prefix $\run \in (\ppdaStates \times \langConcat{\abStackNoBot^*}{\botStack})^*$ of a run, and that is closed under complement and countable union.
The sets in $\mathcal{F}$ are called \emph{measurable} and there is a unique probability measure $\PP_{\ppda} \colon \mathcal{F} \to [0,1]$ satisfying
\[
    \PP_{\ppda}(\langConcat{\run}{\Runs_\ppda})
    ~=~
    \begin{cases}
        \prod_{i = 0}^{|\run|-2} P_{\ppda}(\run(i), \run(i{+}1)) & \text{ if } |\run| = 0 \text{ or } (|\run| > 0 \text{ and } \run(0) = q_0\botStack), \\
        0 & \text{ otherwise},
    \end{cases}
\]
where an empty product (which occurs if $|\run| \leq 1$) is defined to be equal to $1$.
We omit the subscript in $\PP_{\ppda}$ whenever $\ppda$ is given by the context.

In the following two remarks, we summarize the technical differences between our pVPA model and existing models in the literature.

\begin{remark}
    Unlike the pPDA from~\cite{esparzaMCPPDA-lics}, our pVPA only generate \emph{infinite} runs, i.e., they do not ``terminate'' when reaching the empty stack.
    Indeed, in our pVPA, the stack can never be empty because the special bottom symbol $\botStack$ cannot be popped.
    We have chosen this semantics for compatibility with CaRet which describes $\omega$-languages by definition.
    Nonetheless, terminating behavior can be easily simulated in our framework by moving to a dedicated sink state once the pVPA attempts to pop $\botStack$ for the first time.
    Another technical difference between our pVPA and the pPDA introduced in~\cite{esparzaMCPPDA-lics} is that in pVPA, only pop transitions can read the stack, whereas in pPDA, all types of transitions can read, and possibly exchange, the current topmost stack symbol.
    We have chosen this definition (which is not a true restriction) for compatibility with VPA as defined in~\cite{vpl}.
\end{remark}

\begin{remark}
    \label{rem:pVPA}
    The visibility restriction of our pVPA is slightly different from the definition given in~\cite{dubslaff} which requires all \emph{incoming} transitions to a state to be of the same type, i.e., call, internal, or return.
    Our definition, on the other hand, imposes the same requirement on the states' \emph{outgoing} transitions.
    We believe that our condition is more natural for pVPA obtained from procedural programs, such as the one in \Cref{fig:exampleProg}.
    In fact, programs where randomness is restricted to \emph{internal} statements such as $x := \texttt{bernoulli}(0.5)$ or $x := \texttt{uniform}(0,3)$ naturally comply with our visibility condition because all call and return states of such programs are deterministic and thus cannot violate visibility.
    However, the alternative condition of~\cite{dubslaff} is not necessarily fulfilled for such programs.
\end{remark}

We can now formally state our main problem of interest:
\begin{definition}[Probabilistic CaRet Model Checking]
    \label{def:caretDecisionProblems}
    Let $\AP$ be a finite set of atomic propositions,
    $\varphi$ be a CaRet formula over $\AP$,
    $\ppda$ be a pVPA with labels from the pushdown alphabet $\ab = 2^\AP \times \typeset$,
    and $\theta \in [0,1] \cap \Rats$.
    The \emph{quantitative CaRet Model Checking problem} is to decide whether
    \[
        \PP(\Set{\run \in \Runs_{\ppda}}{\labeling(\run) \in \lang{\varphi}}) ~\geq_{?}~ \theta ~.
    \]
    The \emph{qualitative} CaRet Model Checking problem is the special case where $\theta = 1$.
    \qeddef
\end{definition}
The probabilities in \Cref{def:caretDecisionProblems} are well-defined as $\omega$-VPL are measurable~\cite{vpGames}.


\section{Model Checking pVPA against Stair-parity DVPA}
\label{sec:modelCheckingDVPA}

In this section, we show that model checking pVPA (\Cref{def:pVPA}) against VPL given in terms of a stair-parity DVPA (\Cref{def:stairParity}) is decidable.
This is achieved by first computing an automata-theoretic product of the pVPA and the DVPA and then evaluating the acceptance condition in the product automaton.

\subsection{Products of Visibly Pushdown Automata}

In general, pushdown automata are not closed under taking products as this would require \emph{two} independent stacks.
However, the visibility conditions on VPA and pVPA ensure that their product is again an automaton with just a single stack because the stack operations (push, internal, or pop) are forced to synchronize.

We now define the product formally.
An \emph{unlabeled} pVPA is a pVPA where the labeling function $\labeling$ and alphabet $\ab$ are omitted.

\begin{definition}[Product $\automProd{\ppda}{\dvpa}$]
    \label{def:prod}
    Let $\ppda = \ppdaInit$ be a pVPA, and $\dvpa = \npdaInitCustom{\npdaStates}{\sInit}{\abStack'}{\botStack}{\npdaTrans}{\ab}$ be a DVPA over pushdown alphabet $\ab$.
    The product of $\ppda$ and $\dvpa$ is the unlabeled pVPA
    \begin{align*}
        \automProd{\ppda}{\dvpa}
        \eeq
        (\, \ppdaStates \times \npdaStates, ~ \pstate{\sInitPpda}{\sInitNpda}, ~ \abStack \times \abStack', ~ \psymb{\botStack}{\botStack}, ~ \ppdaTransFun_{\ppda\times \dvpa} \,) ~,
    \end{align*}
    where $\ppdaTransFun_{\ppda\times \dvpa}$ is the smallest set of transitions satisfying the following rules for all $q, r \in \ppdaStates$, $Z \in \abStack$, $s,t \in \npdaStates$, and $Y \in \abStack'$:
    \begin{gather*}
        {\small(\text{call})}~\frac{
            \transCallAutom{q}{p}{\ppda}{r}{Z} ~\land~ \transCallAutom{s}{\labeling(q)}{\dvpa}{t}{Y}
        }{
            \transCallAutom{\pstate{q}{s}}{p}{\automProd{\ppda}{\dvpa}}{\pstate{r}{t}}{\psymb{Z}{Y}}
        }
        \quad
        {\small(\text{return})}~\frac{
            \transRetAutom{q}{Z}{p}{\ppda}{r} ~\land~ \transRetAutom{s}{Y}{\labeling(q)}{\dvpa}{t}
        }{
            \transRetAutom{\pstate{q}{s}}{\psymb{Z}{Y}}{p}{\automProd{\ppda}{\dvpa}}{\pstate{r}{t}}
        }
        \\
        {\small(\text{internal})}~\frac{
            \transIntAutom{q}{p}{\ppda}{r} ~\land~ \transIntAutom{s}{\labeling(q)}{\dvpa}{t}
        }{
            \transIntAutom{\pstate{q}{s}}{p}{\automProd{\ppda}{\dvpa}}{\pstate{r}{t}}
        }
        ~.      
    \end{gather*}
    If the DVPA $\dvpa$ is equipped with a priority function $\prioFun \colon \npdaStates \to \Nats$, then we extend $\prioFun$ to $\prioFun' \colon \ppdaStates \times \npdaStates \to \Nats$ via $\prioFun'(q, s) = \prioFun(s)$.
    \qeddef
\end{definition}
It is not difficult to show that $\ppda\times\dvpa$ is indeed a well-defined pVPA and moreover satisfies the following property (the proof is standard, see~\cite[Appendix B.1]{arxivV2}):

\begin{restatable}[Soundness of $\automProd{\ppda}{\dvpa}$]{lemma}{restateProdCorrect}
    \label{thm:product}
    Let $\ppda$ be a pVPA and $\dvpa$ be a stair-parity DVPA with priority function $\prioFun$, both over pushdown alphabet $\ab$.
    Then the product pVPA $\ppda \times \dvpa$ with priority function $\prioFun'$ as in \Cref{def:prod} satisfies
    \begin{align*}
        \PP(\Set{\run \in \Runs_\ppda }{ \lambda(\run)\in \lang{\dvpa}})
        \eeq
        \PP(\Set{\run \in \Runs_{\automProd{\ppda}{\dvpa}}}{ \stepRestr{\run} \in \accSet{\prioFun'}}),
    \end{align*}
    where $\accSet{\prioFun'}$ denotes the set of words in $(\ppdaStates \times \npdaStates)^\omega$ satisfying the \emph{standard} parity condition induced by $\prioFun'$.
    Moreover, $\ppda \times \dvpa$ can be constructed in polynomial time.
\end{restatable}

\begin{remark}
    It is \emph{not} actually important that the product pVPA again satisfies the visibility condition, even though this happens to be the case.
    All techniques we apply to the product also work for general pPDA.
\end{remark}

\subsection{Stair-parity Acceptance Probabilities in pVPA}

\Cref{thm:product} effectively reduces model checking pVPA against stair-parity DVPA to computing stair-parity acceptance in the product, which is again an (unlabeled) pVPA.
We therefore focus on pVPA in this section and do not consider DVPA.

Throughout the rest of this section, let $\ppda = (\, \ppdaStates, \, q_0, \, \abStack, \, \botStack, \, \ppdaTransFun \,)$ be an unlabeled pVPA.
On the next pages we describe the construction of a \emph{finite} Markov chain $\stepMC{\ppda}$ that we call the \emph{step chain} of $\ppda$.
Loosely speaking, $\stepMC{\ppda}$ simulates jumping from one \emph{step} (see \Cref{def:step}) of a run of $\ppda$ to the next.

\begin{remark}
    \label{sec:literatureRemarks}
    The idea of the step Markov chain $\stepMC{\ppda}$ first appeared in~\cite{esparzaMCPPDA-lics}.
    However, the step chain as presented here differs from the original definition in~\cite{esparzaMCPPDA-lics} in at least two important aspects.
    First, we have to take the semantics of our special bottom symbol $\botStack$ into account.
    This is why our chain uses a subset of $\ppdaStates \cup \ppdaStates\botStack$ as states---it must distinguish whether a step occurs at a bottom configuration.
    The pPDA in~\cite{esparzaMCPPDA-lics}, on the other hand, may have both finite and infinite runs, and this needs to be handled differently in the step chain.
    Second, we use step chains for a different purpose than~\cite{esparzaMCPPDA-lics}, namely to show that general measurable properties defined on steps---this includes stair-parity---can be evaluated on pVPA (\Cref{thm:stepChainSound}).
\end{remark}

\subsubsection{Steps as events.}

For all $n \in \Nats$, we define a \emph{random} variable $\V{}{n}$ on $\Runs_{\ppda}$ whose value is either the state $q$ of $\ppda$ at the $n$-th step, or the \emph{extended state} $q\botStack$ in the special case where the $n$-th step occurs at a bottom configuration $q\botStack$, for some $q \in \ppdaStates$.
We denote the set of all such extended states with $\ppdaStates \botStack = \Set{q \botStack}{q \in \ppdaStates}$.
Formally, $\V{}{n} \colon \Runs_{\ppda} \to \ppdaStates \cup \ppdaStates\botStack$ is defined as
\[
    \V{}{n}(\run) \qeq \begin{cases}
        q & \text{ if } \stepi{n}{\run} = q \stack \text{ and } \stack \neq \botStack\\
        q\botStack & \text{ if } \stepi{n}{\run} = q\botStack ~,\\
    \end{cases}
\]
where $\stepi{n}{\run}$ denotes the configuration at the $n$-th step of $\run$.

\begin{restatable}{lemma}{restateVwellDef}
    \label{thm:VwellDef}
    For all $n \in \Nats$ and $v \in \ppdaStates \cup \ppdaStates\botStack$, the event $\V{}{n} = v$ is measurable, and thus $\V{}{n}$ is a well-defined random variable.
\end{restatable}
\begin{proof}
    This was proved more generally in~\cite{esparzaMCPPDA-lics}.
    Here we give an alternative proof using the fact that all $\omega$-VPL are measurable~\cite{vpGames}.
    We view $\ppdaStates \cup \ppdaStates\bot$ as a pushdown alphabet (the partition is induced by the partition on $\ppdaStates$).
    We can construct a non-deterministic Büchi VPA that accepts a word $\word \in (\ppdaStates \cup \ppdaStates\bot)^\omega$ iff the $n$-th step of $\word$ is $v$ (the size of this automaton depends on $n$).
    To this end, the VPA guesses the first $n$ positions that are steps and goes to an accepting state $s$ if the $n$-th step was $v$.
    The automaton can verify that it guessed correctly as follows.
    If it believes a call symbol is a step, it pushes a special marker on the stack; if this marker is ever popped, then the call was no step and the guess was wrong.
    If it detects a wrong guess in this way, then it leaves the accepting state $s$, otherwise it loops there forever.
    The claim follows because the function $f \colon \Runs_\ppda \to (\ppdaStates \cup \ppdaStates\bot)^\omega$ that maps runs to the sequence of their (extended) states is measurable (indeed, the preimage $f^{-1}(\word)$ of every $\word \in (\ppdaStates \cup \ppdaStates\bot)^\omega$ is
    \[
        f^{-1}(\word)
        ~=~
        \bigcap_{i \geq 0} ~
        \bigcup_{\substack{\stack_0 \ldots \stack_i \,\in \, \langConcat{\abStackNoBot^*}{\botStack}} } \langConcat{\word(0)\stack_0 \,\ldots\, \word(i)\stack_i}{\Runs_{\ppda}}
    \]
    which is a countable intersection of countable unions of basic cylinder sets).
\end{proof}

We can view the sequence $\V{}{0},\V{}{1} \ldots$ of random variables as a \emph{stochastic process}.
It is intuitively clear that for all $n \in \Nats$, the value of $\V{}{n+1}$ depends only on $\V{}{n}$, but not on $\V{}{i}$ for $i < n$.
This is due to the more general observation that only the \emph{state} $q$ at any step configuration $q\stack$ (with $\stack \neq \botStack$) \emph{fully determines the future of the run} because being a step already implies that no symbol in $\stack$ can ever be read, as reading it implies popping it from the stack.
In particular, $q$ determines the probability distribution over possible next steps.
A similar observation applies to bottom configurations of the form $q \botStack$.
Phrased in the language of probability theory, the process $\V{}{0},\V{}{1} \ldots$ has the \emph{Markov property}, i.e.,
\begin{equation}
    \label{eq:markov_prop}
    \PP(\V{}{n} = v_n \mid \V{}{n-1} = v_{n-1} \,\wedge\, \ldots \,\wedge\, \V{}{0} = v_0)
    ~=~
    \PP(\V{}{n} = v_n \mid \V{}{n-1} = v_{n-1})
\end{equation}
holds for all values of $v_0,\ldots,v_n$ such that the conditional probabilities are well-defined\footnote{
    A conditional probability is well-defined if the \emph{condition}, i.e., the event on the right hand side of the vertical bar, has positive probability.
    Expressions like the one in \eqref{eq:markov_prop} are thus not necessarily well-defined because the probability that $\V{}{n-1} = v_{n-1}$ might be zero for certain values of $n$ and $v_{n-1}$.
}.
This was proved in detail in~\cite{esparzaMCPPDA-lics}.
It is also clear that the Markov process is time-homogeneous in the sense that
\begin{align}
    \label{eq:timehom}
    \condProb{\V{}{n+1} = v'}{\V{}{n} = v} \qeq \condProb{\V{}{m+1} = v'}{\V{}{m} = v}
\end{align}
holds for all $n, m \in \Nats$ for which the two conditional probabilities are well-defined.
The following example provides some intuition on these facts.
\begin{example}
    \label{ex:stepChainSimple}
    Consider the pVPA in \Cref{fig:step-mc-simple} (left).
    The initial fragments of its two equiprobable runs are depicted in the middle.
    In this example, it is easy to read off the next-step probabilities $\PP(\V{}{n} = v_n \mid \V{}{n-1} = v_{n-1})$ for all $n \in \Nats$ and $v_n, v_{n-1} \in \ppdaStates \cup \ppdaStates\botStack$.
    They are summarized in the Markov chain on the right.
    For example, $\V{}{0} = \sInitPpda\botStack$ holds with probability 1, and $\V{}{1} = q_1$ and $\V{}{1} = q_3\botStack$ hold with probability $\nicefrac 1 2$ each because the second step occurs either at position $1$ with configuration $q_1 \botStack Z$ or at position $3$ with configuration $q_3 \botStack$, and both options are equally likely.
    The case $\PP(\V{}{2} = q_2 \mid \V{}{1} = q_1) = 1$ is slightly more interesting:
    Given that a configuration $q_1 \gamma$ with $\gamma \neq \botStack$ is a step, we know that the next state must be $q_2$ (which is then also a step).
    Even though there is a transition from $q_1$ to $q_3$ in $\ppda$, the next state cannot be $q_3$ because the latter is a return state which would immediately decrease the stack height of $\gamma$. 
    This shows that, intuitively speaking, \emph{conditioning on being a step influences the probabilities of a state's outgoing transitions}.
\end{example}

\begin{figure}[t]
    \centering
    \begin{minipage}{0.44\linewidth}    
        \begin{tikzpicture}[myautomaton, initial text=, initial where=above,node distance=8mm and 18mm,on grid]
        \node[callState,initial] (q0) {$q_0$};
        \node[state,right=of q0] (q1) {$q_1$};
        \node[callState,above right=of q1,rectangle] (q2) {$q_2$};
        \node[retState,below right=of q1] (q3) {$q_3$};
        \draw[->] (q0) -- node[above] {$1, Z$} (q1);
        \draw[->] (q1) -- node[auto] {$\nicefrac 1 2$} (q2);
        \draw[->] (q1) -- node[below left] {$\nicefrac 1 2$} (q3);
        \draw[->] (q2) edge[loop right] node[right] {$1, Z$} (q2);
        \draw[->] (q3) edge[loop right] node[right,align=right] {$Z, 1$ \\[-2pt] $\botStack, 1$} (q3);
        \end{tikzpicture}
    \end{minipage}
    \hfill
    \begin{minipage}{0.21\linewidth}
        {\renewcommand{\arraystretch}{0.7}
         \setlength{\tabcolsep}{2pt}
            \begin{tabular}{c c c c c c}
                \tiny 0 &   \tiny 1 &  \tiny 2 & \tiny  3&   \tiny 4& \tiny 5 \\
                &   &   &   &   & ...\\
                &   &   &   & $Z$ & ...\\
                &   &   & $Z$ & $Z$ & ...\\
                & $Z$ & $Z$ & $Z$ & $Z$ & ...\\
                $\botStack$ & $\botStack$ & $\botStack$ & $\botStack$ & $\botStack$ & ...\\[3pt]
                $\underline{q_0}$ & $\underline{q_1}$ & $\underline{q_2}$ & $\underline{q_2}$ & $\underline{q_2}$ & ... \\
            \end{tabular}
            \vspace{5mm}\\
            \begin{tabular}{c c c c c c}
                & $Z$ & $Z$ &  &  & \\
                $\botStack$ & $\botStack$ & $\botStack$ & $\botStack$ & $\botStack$ & ...\\[3pt]
                $\underline{q_0}$ & $q_1$ & $q_3$ & $\underline{q_3}$ & $\underline{q_3}$ & ... \\
            \end{tabular}
        }
    \end{minipage}
    \hfill
    \begin{minipage}{0.23\linewidth}
        \begin{tikzpicture}[myautomaton, initial text=, initial where=left,on grid,node distance=15mm and 14mm]
        \tikzstyle{myState} = [state,inner sep=1pt]
        \node[myState,initial] (q0bot) {$q_0\botStack$};
        \node[myState,below=of q0bot] (q3bot) {$q_3\botStack$};
        \node[myState,above=of q0bot] (q1) {$q_1$};
        \node[myState,right=of q1] (q2) {$q_2$};
        \draw[->] (q0bot) -- node[right] {$\nicefrac 1 2$} (q3bot);
        \draw[->] (q0bot) -- node[right] {$\nicefrac 1 2$} (q1);
        \draw[->] (q1) -- node[above] {$1$} (q2);
        \draw[->] (q2) edge[loop below] node[below] {$1$} (q2);
        \draw[->] (q3bot) edge[loop right] node[right] {$1$} (q3bot);
        \end{tikzpicture}
    \end{minipage}
    \caption{
        Left: An example (unlabeled) pVPA $\ppda$.
        Recall that call and return states are drawn as squares and rhombs, respectively,  whereas internal states are depicted as circles.
        Middle: Initial fragments of the two possible runs of $\ppda$.
        Steps are underlined.
        Right: The step Markov chain $\stepMC{\ppda}$ (\Cref{def:stepMC}, page~\pageref{def:stepMC}).
    }
    \label{fig:step-mc-simple}
\end{figure}

\subsubsection{Probabilities of next steps, returns, and diverges.}

Our next goal is to provide expressions for the next-step probabilities $\condProb{\V{}{n+1} = v'}{\V{}{n} = v}$ as we did in \Cref{ex:stepChainSimple}.
It turns out that those can be stated in terms of the \emph{return} and \emph{diverge} probabilities of $\ppda$.

\begin{definition}
    \label{def:returnDiverge}
    Let $q,r \in Q$ be states, and $Z \in \abStackNoBot$ be a stack symbol of pVPA $\ppda$.
    We define the following probabilities:
    \begin{itemize}
        \item The \emph{return} probability $\termProbFromTo{qZ}{r}$ is the probability to reach configuration $r \botStack$ from $q Z \botStack $ without visiting another bottom configuration $t \botStack$ for some $t \neq r$ in between.
        Formally,
        \[
            \termProbFromTo{qZ}{r}
            ~=~
            \PP_{q Z \botStack}(\, \Set{q' \stack \botStack}{q' \in Q, \stack \in \abStackNoBot^+ } ~\until~ \{r\botStack\} \,)
        \]
         where $\PP_{q Z \botStack}$ is the probability measure associated with the infinite Markov chain $\ppdaMcSemantic{\ppda}$ assuming initial state $q Z \botStack$, and $\until$ is the standard until operator from LTL.
        \item The \emph{diverge} probability $\nonTermProb{q} = 1 - \sum_{t\in Q} \termProbFromTo{qZ}{t}$, i.e., the probability to never pop $Z$ from the stack when starting in $q Z \botStack$.
        Note that $\nonTermProb{q}$ is indeed independent of $Z$ because the only way to read $Z$ is by popping it from the stack. Recall that this is due to our specific definition of pVPA (\Cref{def:pVPA}) in which only pop transitions can read from the stack just like in VPA (\Cref{def:vpa}); we remark that in traditional (p)PDA, all types of transition can read ---and possibly replace--- the topmost stack symbol~\cite{esparzaMCPPDA-lics}.
        \qeddef
    \end{itemize}
\end{definition}

The diverge probabilities are closely related to steps.
In fact, the probability that a non-bottom configuration with head $q Z$ is a step is equal to $\nonTermProb{q}$.
For example, in the pVPA in \Cref{fig:step-mc-simple} the configuration $q_1 Z \bot $ is a step with probability $\nonTermProb{q_1} = \nicefrac{1}{2}$.

\begin{example}
    \label{ex:nonRationalProbs}
    It is well known that the return and diverge probabilities are not necessarily rational.
    We give a minimal example to illustrate this fact.
    Consider the following pVPA:
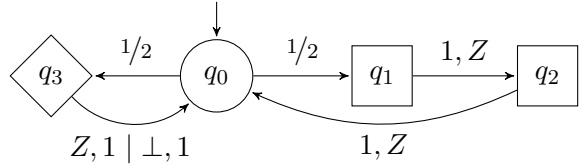
    \begin{center}
        \begin{tikzpicture}[myautomaton, initial text=, initial where=above,node distance=8mm and 22mm,on grid]
            \node[state,initial] (q0) {$q_0$};
            \node[callState,right=of q0] (q1) {$q_1$};
            \node[callState,right=of q1] (q2) {$q_2$};
            \node[retState,left=of q0] (q3) {$q_3$};
            \draw[->] (q0) -- node[above] {$\nicefrac 1 2$} (q1);
            \draw[->] (q1) -- node[auto] {$1, Z$} (q2);
            \draw[->] (q0) edge[bend right=0] node[above] {$\nicefrac 1 2$} (q3);
            \draw[->] (q2) edge[bend left=25] node[auto] {$1, Z$} (q0);
            \draw[->] (q3) edge[bend right=45] node[below] {$Z, 1 \mid \botStack, 1$} (q0);
        \end{tikzpicture}
    \end{center}
    Intuitively, this pVPA either pops the topmost symbol with probability $\nicefrac 1 2$, or it pushes two times the symbol $Z$.
    Note that all return probabilities of the form $\termProbFromTo{\ldots}{q_i}$ for $i \neq 0$ are equal to zero.
    In~\cite{esparzaMCPPDA-lics} it was shown that the remaining return probabilities are the \emph{component-wise least non-negative} solution of the \emph{polynomial system}:
    \begin{align*}
        \termProbFromTo{q_0Z}{q_0} &~=~ \frac{1}{2} \cdot \termProbFromTo{q_1Z}{q_0} + \frac{1}{2} \cdot \termProbFromTo{q_3Z}{q_0} \\
        \termProbFromTo{q_1Z}{q_0} &~=~ \termProbFromTo{q_2Z}{q_0} \cdot \termProbFromTo{q_0Z}{q_0} \\
        \termProbFromTo{q_2Z}{q_0} &~=~ \termProbFromTo{q_0Z}{q_0} \cdot \termProbFromTo{q_0Z}{q_0} \\
        \termProbFromTo{q_3Z}{q_0} &~=~ 1 ~.
    \end{align*}
    It follows that $\termProbFromTo{q_0Z}{q_0}$ must be the least non-negative solution of
    \begin{align*}
        \termProbFromTo{q_0Z}{q_0} ~=~ \frac{1}{2} \cdot \termProbFromTo{q_0Z}{q_0}^3 + \frac{1}{2}
    \end{align*}
    which is $\termProbFromTo{q_0Z}{q_0} = \tfrac{\sqrt{5} - 1}{2} \approx 0.618$, the reciprocal of the golden ratio.
    \qeddef
\end{example}

The probabilities in \Cref{ex:nonRationalProbs} can still be expressed by radicals (square roots, cubic roots, and so on) which allows for certain effective computations.
However, in general, the probabilities cannot even be expressed in this way.
For example, consider a pVPA that repeats the following steps until emptying its stack or getting stuck in a sink state: (i) It pushes four symbols with probability $\tfrac 1 6$, or (ii) pops one symbol with probability $\tfrac 1 2$, or (iii) gets stuck otherwise.
The resulting return probability is the least $x \geq 0$ with $x = \tfrac 1 6 x^5 + \tfrac 1 2$, which is an algebraic number not solvable by radicals~\cite[Theorem 3.2(1)]{DBLP:journals/jacm/EtessamiY09}.

\begin{remark}
    The probabilities $\termProbFromTo{qZ}{r}$ from \Cref{def:returnDiverge} were called \emph{termination probabilities} in previous work~\cite{brazdilSurvey}.
    However, we believe that \emph{return} probability is more appropriate.
    When modeling procedural probabilistic programs as pVPA, $\termProbFromTo{qZ}{r}$ is just the probability to eventually \emph{return} from local state $q$ of the current procedure to local state $r$ of the calling procedure (the return address is stored on the stack in $Z$).
    We believe that the term \emph{termination} probability is more adequate for referring to the quantity $\sum_{r \in \ppdaStates} \termProbFromTo{q_0Z_0}{r}$, where $Z_0$ is some initial stack symbol, i.e., the probability that some initial procedure indentified by $Z_0$ returns at all when started in local state $q_0$.
\end{remark}

We now state the technical key lemma of this section, the characterization of the next step probabilities $\condProb{\V{}{n+1} = v'}{\V{}{n} = v}$ as given in  \Cref{fig:transProbs}.
The upcoming section is devoted to proving that the entries in the table are correct.

\begin{restatable}{lemma}{restateTransProbs}
    \label{thm:transProbs}
    The conditional next-step probabilities in \Cref{fig:transProbs} are correct in the sense that if $\condProb{\V{}{n+1} = v'}{\V{}{n} = v}$ is defined for $n \in \Nats$ and $v, v' \in Q \cup Q\bot$, then it is equal to the probability in the respective column $v \to v'$.
\end{restatable}

\newcommand{\transProbTable}{
    {\renewcommand{\arraystretch}{1.1}
        \begin{adjustbox}{max width=\textwidth}
            \begin{tabular}{l  c  c  c  c }
                \toprule
                & $q \to r$& $q \bot \to r $ & $q \bot \to r \bot$ & $q \to r \bot$\\
                \midrule
                $q \in \ppdaStatesCall ~$
                & $\displaystyle \frac{ \nonTermProb{r} } {\nonTermProb{q}} \bigg( \sum_{r', Z} \ppdaTransFun_{\callSubscript}(q, r'Z)\termProbFromTo{r'Z}{r} + \sum_{Z} \ppdaTransFun_{\callSubscript}(q,rZ) \bigg)$
                & $\displaystyle \sum_{Z} \ppdaTransFun_{\callSubscript}(q,rZ) \nonTermProb{r} \quad$
                & $\displaystyle\sum_{r', Z} \ppdaTransFun_{\callSubscript}(q, r'Z) \termProbFromTo{r'Z}{r}$
                & 0\\
                \midrule
                $q \in \ppdaStatesInt$
                & $\displaystyle \frac{\nonTermProb{r}}{\nonTermProb{q}} \ppdaTransFun_{\intSubscript}(q, r) $
                & $0$
                & $\ppdaTransFun_{\intSubscript}(q,r)$
                & 0\\
                \midrule
                $q \in \ppdaStatesRet$
                & undef.
                & 0
                & $\ppdaTransFun_{\retSubscript}(q\botStack, r)$
                & undef.\\\bottomrule
            \end{tabular}
        \end{adjustbox}
    }
}
\begin{table}[t]   
    \centering
    \caption{
        Next-step probabilities of the step Markov chain.
        $\ppdaTransFun_{\mathsf{type}}$ for $\mathsf{type} \in \typeset$ are the probabilities of the pVPA's call, internal, and return transitions, respectively.
        The values $\termProbFromTo{r'Z}{r}$ and $\nonTermProb{q}$ are the return and diverge probabilities from \Cref{def:returnDiverge}.
    }
    \label{fig:transProbs}
    \transProbTable
\end{table}

\subsubsection{Proof of \Cref{thm:transProbs}}

We first explain the trivial entries in \Cref{fig:transProbs}.
Further below, we give a self-contained proof of the two non-trivial expressions in the left-most column of \Cref{fig:transProbs}.
Throughout the whole proof we fix an (unlabed) pVPA $\ppda = ( \ppdaStates, \, q_0, \, \abStack, \, \botStack, \, \ppdaTransFun)$, with $\ppdaTransFun = (\ppdaTransFun_{\callSubscript}, \ppdaTransFun_{\intSubscript}, \ppdaTransFun_{\retSubscript})$ the call, internal, and return transition functions, respectively.
The following items correspond to the trivial entries in \Cref{fig:transProbs} and are ordered column-by-column, from left to right:

\begin{itemize}
    \item The probability $\condProb{\V{}{n+1} = r}{\V{}{n} = q}$ with $q \in \ppdaStatesRet$ is never well-defined because it is impossible that steps occur at non-bottom configurations with a return state.
    \item The probability $\condProb{\V{}{n+1} = r}{\V{}{n} = q\bot}$ with $q \in \ppdaStatesInt$ is trivially zero because if $q$ is internal then the next step after a bottom configuration $q\bot$ is necessarily also a bottom configuration.
    \item The probability $\condProb{\V{}{n+1} = r}{\V{}{n} = q\bot}$ with $q \in \ppdaStatesRet$ is trivially zero because if $q$ is a return state at a bottom configuration then the next step occurs at the immediate successor configuration which is a bottom configuration as well.
    \item The probability $\condProb{\V{}{n+1} = r\bot}{\V{}{n} = q\bot}$ with $q \in \ppdaStatesInt$ is straightforward because $q\bot$ and $r\bot$ are both steps and the probability that the next state after $q\bot$ is $r\bot$ is $\ppdaTransFun_{\intSubscript}(q,r)$.
    \item The probability $\condProb{\V{}{n+1} = r\bot}{\V{}{n} = q\bot}$ with $q \in \ppdaStatesRet$ is simply $\Pret(q\bot, r)$ for the same reason as in the previous case (recall that a return state at a bottom-configuration behaves exactly like an internal one).
    \item All the remaining probabilities in the rightmost column ``\,$q \to r\bot$\,'' are trivially zero or ill-defined because if a step occurs at non-bottom configuration, then the next step can never occur at a bottom configuration.
\end{itemize}

We now focus on the following non-trivial cases.
Let $r \in \ppdaStates$ and $n \in \Nats$ be arbitrary.
\begin{enumerate}
    \item If $q \in \ppdaStatesInt$ then,
    \[
        \condProb{\V{}{n+1} = r}{\V{}{n} = q}
        \eeq
        \frac{\nonTermProb{r}}{\nonTermProb{q}} \ppdaTransFun_{\intSubscript}(q,r)
        ~.
    \]
    \label{it:eq1}
    \item  If $q \in \ppdaStatesCall$ then,
    \[
        \condProb{\V{}{n+1} = r}{\V{}{n} = q}
        \eeq
        \frac{ \nonTermProb{r} } {\nonTermProb{q}} \left( \sum_{r', Z} \ppdaTransFun_{\callSubscript}(q, r'Z)\termProbFromTo{r'Z}{r} + \sum_{Z} \ppdaTransFun_{\callSubscript}(q,rZ)\right) ~.
    \]
    \label{it:eq2}
\end{enumerate}
The other two non-trivial cases are easier variants of case (\ref{it:eq2}), hence we omit them here (see~\cite[p. 30]{arxivV2} for details).
Next we provide some intuition about cases (\ref{it:eq1}) and (\ref{it:eq2}):
\begin{itemize}
    \item 
    For (\ref{it:eq1}), suppose that the $n$-th step is at position $i$ of the run.
    Since the $n$-th step occurs at an \emph{internal} state $q \in \ppdaStatesInt$, the $n{+}1$-st step must necessarily occur immediately at position $i{+}1$.
    The factor $\Pint(q,r) \nonTermProb{r}$ is proportional to the probability to take an (internal) transition from $q$ to $r$ and then diverge in $r$, which is necessary in order for the next configuration to be a step.
    However, the values $\Set{\Pint(q,r) \nonTermProb{r}}{r \in Q}$ do not form a probability distribution in general.
    Therefore we divide by the normalizing constant $\nonTermProb{q} = \sum_{r \in Q} \Pint(q,r) \nonTermProb{r}$.
    \item 
    In (\ref{it:eq2}), the two summands correspond to the following case distinction:
    Since the $n$-th step occurs at the \emph{call} state $q \in \ppdaStatesCall$, the $n{+}1$-st step either (i) occurs at the same stack height as the current step $n$ (which means that the current call has a matching return), or (ii) the stack height at the next step $n{+}1$ is incremented by 1 compared to the stack height at step $n$.
    In case (ii), the next step occurs immediately at the next position, which is why the second summand is just the 1-step probability to go from $q$ to $r$.
    In case (i), the symbols pushed by the outgoing transitions of $q$ must be eventually popped.
    For instance, if we assume that $q$ takes a transition to an immediate successor $r'$ and pushes $Z$ on the stack, then the probability that the next step occurs at $r$ is precisely the \emph{return probability} $\termProbFromTo{r'Z}{r}$ (see~\Cref{def:returnDiverge}).
    The factor $\tfrac{\nonTermProb{r}}{\nonTermProb{q}}$ is needed for similar reasons as in (\ref{it:eq1}).
\end{itemize}

We now give the formal proofs for \eqref{it:eq1} and \eqref{it:eq2}.
In the following, we often use equations involving conditional probabilities such as $\condProb{A}{B} = \condProb{C}{D}$.
These conditional probabilities are not necessarily well-defined in all cases.
Therefore, the meaning of our equations is that they hold only if all probabilities involved are well-defined.
We need some definitions and (simple) lemmas first.

\begin{definition}
    Let $i \in \Nats$ and $q \in \ppdaStates$.
    We introduce the following events:
    \begin{itemize}
        \item $\stateAtPos{q}{i}$ is the set of all runs $\run \in \Runs_{\ppda}$ such that $\run(i) = q \stack$ with $\stack \neq \botStack$, i.e., the runs whose $i$-th configuration has state $q$ and stack unequal to $\bot$.
        \item Similarly, $\stateAtPos{q\botStack}{i}$ denotes the set of all runs $\run \in \Runs_{\ppda}$ such that $\run(i) = q \bot$, i.e., the runs whose $i$-th configuration is a \emph{bottom} configuration with state $q$.
        \item $\stepAtPos{i}$ denotes the set of all runs such that the $i$-th configuration is a step.
        \item We define $\stackAt{i} = \sh{\run(i)} \in \Nats$, i.e., the stack height of the $i$-th configuration. Strictly speaking, $\stackAt{i}$ is a random variable, not an event.
        Note that $\stepAtPos{i}$ is by definition equivalent to $\forall j > i \colon \stackAt{j} \geq \stackAt{i}$.
    \end{itemize}
    These events are all measurable.
    \qeddef
\end{definition}

\begin{lemma}
    For all $i \in \Nats$, and $q \in \ppdaStates$, the following identities hold:
    \begin{align}
        \condProb{\stepAtPos{i}}{\stateAtPos{q}{i}} &\eeq \nonTermProb{q}
        \label{eq:stepProbIsNonTermProb} \\
        \condProb{\stepAtPos{i}}{\stateAtPos{q\botStack}{i}} &\eeq 1
        \label{eq:botIsTriviallyStep}
    \end{align}
    Further, for $q \in \ppdaStatesInt$ and $r \in \ppdaStates$,
    \begin{align}
        \condProb{\stateAtPos{r}{(i{+}1)}}{\stateAtPos{q}{i}} &\eeq \ppdaTransFun_{\intSubscript}(q,r) \label{eq:nextPosInt} \\
        \condProb{ \stepAtPos{(i{+}1)} \land \stepAtPos{i}}{\stateAtPos{r}{(i{+}1)} \land \stateAtPos{q}{i}} & \eeq \condProb{ \stepAtPos{(i{+}1)} }{\stateAtPos{r}{(i{+}1)} } \label{eq:internalSimplification}
    \end{align}
\end{lemma}
\begin{proof}
    The first three equations follow immediately from the definitions.
    For \eqref{eq:internalSimplification} we have:
    \begin{align*}
        & \condProb{ \stepAtPos{(i{+}1)} \land \stepAtPos{i}}{\stateAtPos{r}{(i{+}1)} \land \stateAtPos{q}{i}} \\
        \qeq & \condProb{ \stepAtPos{(i{+}1)} }{\stateAtPos{r}{(i{+}1)}\land \stateAtPos{q}{i} } \\
        \qeq & \condProb{ \stepAtPos{(i{+}1)} }{\stateAtPos{r}{(i{+}1)} }
    \end{align*}
    The first equation holds because if $q \in \ppdaStatesInt$ and $\stateAtPos{q}{i}$, then $\stepAtPos{(i{+}1)} \land \stepAtPos{i} $ is already implied by $\stepAtPos{(i{+}1)}$, and the second equation holds because the probability that $\stepAtPos{(i{+}1)}$ depends only on the state at position $i{+}1$, not on the state at position $i$.
\end{proof}

To prove equation for case (\ref{it:eq1}) we argue as follows.
By time-homogeneity (see \eqref{eq:timehom}) and the definition of $\V{}{n}$, we have for all $i,n \in \Nats$, $q \in \ppdaStatesInt$ and $r \in \ppdaStates$ that
\begin{align}
    \condProb{\V{}{n+1} = r}{\V{}{n} = q} &\eeq \condProb{\stateAtPos{r}{(i{+}1)} \land \stepAtPos{(i{+}1)}}{\stateAtPos{q}{i} \land \stepAtPos{i}} \label{eq:intCol1Ident}
\end{align}
Now:

\begin{align*}
    &\condProb{\V{}{n+1} = r}{\V{}{n} = q} \\
    =\quad & \condProb{\stateAtPos{r}{(i{+}1)} \land \stepAtPos{(i{+}1)}}{\stateAtPos{q}{i} \land \stepAtPos{i}}
    \tag{by \eqref{eq:intCol1Ident}}\\
    =\quad &\frac{\PP(\stateAtPos{r}{(i{+}1)} \land \stepAtPos{(i{+}1)} \land \stateAtPos{q}{i} \land \stepAtPos{i})}{\PP(\stateAtPos{q}{i} \land \stepAtPos{i})}
    \tag{cond.\ probability}\\
    =\quad &\frac{\condProb{\stateAtPos{r}{(i{+}1)} \land \stepAtPos{(i{+}1)} \land \stepAtPos{i}}{\stateAtPos{q}{i}} \cdot \PP(\stateAtPos{q}{i})}  {\condProb{\stepAtPos{i}}{\stateAtPos{q}{i}} \cdot \PP(\stateAtPos{q}{i})}
    \tag{cond.\ probability}\\
    =\quad &\frac{\condProb{\stateAtPos{r}{(i{+}1)} \land \stepAtPos{(i{+}1)} \land \stepAtPos{i}}{\stateAtPos{q}{i}}}{\condProb{\stepAtPos{i}}{\stateAtPos{q}{i}}} 
    \tag{simplification}\\
    =\quad &\frac{\condProb{ \stepAtPos{(i{+}1)} \land \stepAtPos{i}}{\stateAtPos{r}{(i{+}1)} \land \stateAtPos{q}{i}} \cdot \condProb{\stateAtPos{r}{(i{+}1)}}{\stateAtPos{q}{i}}}{\condProb{\stepAtPos{i}}{\stateAtPos{q}{i}}}
    \tag{cond.\ probability}\\
    =\quad &\frac{\condProb{ \stepAtPos{(i{+}1)} }{\stateAtPos{r}{(i{+}1)} } \cdot \condProb{\stateAtPos{r}{(i{+}1)}}{\stateAtPos{q}{i}}}{\condProb{\stepAtPos{i}}{\stateAtPos{q}{i}}}
    \tag{by \eqref{eq:internalSimplification}}\\
    =\quad &\frac{\nonTermProb{r} \cdot \ppdaTransFun_{\intSubscript}(q,r)}{\nonTermProb{q}} 
    \tag{by \eqref{eq:stepProbIsNonTermProb}, \eqref{eq:nextPosInt}}~.
\end{align*}
This concludes the proof for case \eqref{it:eq1}.

We now turn to case (\ref{it:eq2}).
For all $i,n \in \Nats$, $q \in \ppdaStatesCall$ and $r \in \ppdaStates$, it holds that
\begin{align}
    &\condProb{\V{}{n+1} = r}{\V{}{n} = q} \notag\\
    \qeq & \condProb{\exists k>i: \stepAtPos{k} \land \stateAtPos{r}{k} \land \forall i<j<k: \neg \stepAtPos{j}}{\stepAtPos{i} \land \stateAtPos{q}{i}} \tag{by time homogeneity and definition of $\V{}{n}$}\\
    \begin{split}
        \qeq & \condProb{\exists k>i{+}1: \stepAtPos{k} \land \stateAtPos{r}{k} \land \forall i<j<k: \neg \stepAtPos{j}}{\stepAtPos{i} \land \stateAtPos{q}{i}} \\
        & +~ \condProb{\stateAtPos{r}{(i{+}1)} \land \stepAtPos{(i{+}1)}}{\stateAtPos{q}{i} \land \stepAtPos{i}} 
    \end{split}
    \label{eq:twosummands}
\end{align}
The last equality results from a split in two disjoint events.
For the second summand in \eqref{eq:twosummands} it can be shown that
\[
    \condProb{\stateAtPos{r}{(i{+}1)} \land \stepAtPos{(i{+}1)}}{\stateAtPos{q}{i} \land \stepAtPos{i}}
    ~=~
    \frac{\nonTermProb{r}}{\nonTermProb{q}} \sum_{Z}   \ppdaTransFun_{\callSubscript}(q,rZ)
\]
by a similar derivation as in case (\ref{it:eq1}) (the sum over all stack symbols $Z$ is because $q \in \ppdaStatesCall$, so that there may be multiple ---up to $|\abStackNoBot|$ many--- direct transitions from $q$ to $r$).

We need a couple of lemmas before deriving an equation for the first summand in \eqref{eq:twosummands}.

\begin{lemma}
    \label{thm:threeEquations}
    For all $q \in \ppdaStatesCall$, $r \in \ppdaStates$, and $i \in \Nats$ it holds that:
    \begin{align}
        \begin{split}
            &\sum_{r', Z} \ppdaTransFun_{\callSubscript}(q, r'Z)\termProbFromTo{r'Z}{r} 
            \eeq  \sum_{k>i+1} \condProb{\stateAtPos{r}{k} \land \forall i{<}j{<}k: \stackAt j > \stackAt k = \stackAt i }{\stateAtPos{q}{i}}
        \end{split}
        \label{eq:returnProbsStackHeights}
    \end{align}
    Moreoever, for all $q \in \ppdaStatesCall$, $i \in \Nats$, and $k \in \Nats$, we have:
    \begin{align}
        \begin{split}
            &\condProb{\stepAtPos{k} \land \forall i{<}j{<}k: \neg \stepAtPos{j} \land \stepAtPos{i} } {\stateAtPos{q}{i}} \\
            \eeq &\condProb{\stepAtPos{k} \land \forall i{<}j{<}k: \stackAt j > \stackAt i= \stackAt k } {\stateAtPos{q}{i}}
        \end{split}
        \label{eq:stepStackHeights}
    \end{align}
    and
    \begin{align}
        \begin{split}
            &\condProb{\forall i{<}j{<}k: \stackAt j > \stackAt i= \stackAt k } {\stateAtPos{q}{i}\land \stateAtPos{r}{k} \land \stepAtPos{k}} \\
             \eeq & \condProb{\forall i{<}j{<}k: \stackAt j > \stackAt i= \stackAt k } {\stateAtPos{q}{i}\land \stateAtPos{r}{k}}
        \end{split}        
        \label{eq:stepIndep}
    \end{align}
\end{lemma}
\begin{proof}
    For \eqref{eq:returnProbsStackHeights}, note that $\sum_{r', Z} \ppdaTransFun_{\callSubscript}(q, r'Z)\termProbFromTo{r'Z}{r}$ is the probability to go from $q$ (with empty stack) to a successor state $r'$, push $Z \in \abStack$ and then later reach $r$ with empty stack within finitely many steps.
    If we assume that $\stateAtPos{q}{i}$, then this is the same as summing over all positions $k>i+1$ (we can exclude $k=i+1$ because is not possible because to push and pop $Z$ within one transition) such that $\stateAtPos{r}{k}$ and for all $i<j<k$ the stack height is greater than at position $i$ and $k$.
    Positions $i$ and $k$ have the same stack height because in the transition from $i$ the symbol $Z$ is pushed, and $k$ is the position directly after $Z$ is popped.
    In between those two transitions, the stack below $Z$ cannot change, so the stack is the same at both positions.
    
    For \eqref{eq:stepStackHeights} we argue as follows:
    \begin{align*}
        & \condProb{\stepAtPos{k} \land \forall i{<}j{<}k: \neg \stepAtPos{j} \land \stepAtPos{i} } {\stateAtPos{q}{i}}\\
        \eeq & \condProb{\stepAtPos{k} \land \forall i{<}j{<}k: \stackAt j > \stackAt i= \stackAt k \land \stepAtPos{i}} {\stateAtPos{q}{i}}\\
        \eeq & \condProb{\stepAtPos{k} \land \forall i{<}j{<}k: \stackAt j{>}\stackAt i= \stackAt k } {\stateAtPos{q}{i}}
    \end{align*}
    The first equality holds because if no position between $i$ and $k$ is a step, then the stack at those positions must be higher than at position $k$.
    Furthermore, since $i$ is a step, we have $\stackAt i \leq \stackAt k$; and moreover, since $i{+}1$ is not a step and $q$ is a call state, we even have $\stackAt i = \stackAt k$.
    The second equality holds because if $i$ and $k$ have the same stack height and all positions between them have a higher stack, then $i$ is a step if and only of $k$ is a step.
    
    Equation \eqref{eq:stepIndep} is somewhat counter-intuitive because conditioning on $\stepAtPos{k}$ is like ``conditioning on the future'':
    The stack height \emph{after} position $k$ should never be smaller than at position $k$.
    Knowing that $\stepAtPos{k}$ gives information about the (extended) state at position $k$.
    However, in \eqref{eq:stepIndep} we also condition on the fact that $\stateAtPos{r}{k}$, i.e., at position $k$, the run is in step $r$ and the topmost stack symbol is not $\botStack$.
    Hence, in the context of \eqref{eq:stepIndep}, we can drop $\stepAtPos{k}$ from the condition.
\end{proof}

We conclude the proof of case \eqref{it:eq2} and thus of the whole \Cref{thm:transProbs} by deriving the desired equation for the first summand in \eqref{eq:twosummands}:
{\allowdisplaybreaks\begin{align*}
    & \condProb{\exists k {>} i{+}1: \stepAtPos{k} \land \stateAtPos{r}{k} \land \forall i{<}j{<}k: \neg \stepAtPos{j}}{\stepAtPos{i} \land \stateAtPos{q}{i}} \\
    = \quad & \frac{ \condProb{\exists k {>} i{+}1: \stepAtPos{k} \land \stateAtPos{r}{k} \land \forall i{<}j{<}k: \neg \stepAtPos{j} \land \stepAtPos{i} } {\stateAtPos{q}{i}} \cdot \PP(\stateAtPos{q}{i})}{\condProb{\stepAtPos{i}}{ \stateAtPos{q}{i}} \cdot \PP(\stateAtPos{q}{i})}
    \tag{cond.\ probability twice}\\
    = \quad & \frac{ \condProb{\exists k {>} i{+}1: \stepAtPos{k} \land \stateAtPos{r}{k} \land \forall i{<}j{<}k: \neg \stepAtPos{j} \land \stepAtPos{i} } {\stateAtPos{q}{i}}}{\condProb{\stepAtPos{i}}{ \stateAtPos{q}{i}}}
    \tag{simplification}\\
    = \quad & \sum_{k > i{+}1} \frac{ \condProb{\stepAtPos{k} \land \stateAtPos{r}{k} \land \forall i{<}j{<}k: \neg \stepAtPos{j} \land \stepAtPos{i} } {\stateAtPos{q}{i}}}{\condProb{\stepAtPos{i}}{ \stateAtPos{q}{i}}}
    \tag{split disjoint events}\\
    = \quad & \sum_{k > i{+}1} \frac{ \condProb{\stepAtPos{k} \land \stateAtPos{r}{k} \land \forall i{<}j{<}k: \stackAt j > \stackAt i= \stackAt k } {\stateAtPos{q}{i}}}{\condProb{\stepAtPos{i}}{ \stateAtPos{q}{i}}}
    \tag{by \eqref{eq:stepStackHeights} }\\
    \begin{split}
    = \quad & \sum_{k > i{+}1} \bigg( \condProb{\stepAtPos{k} \land \stateAtPos{r}{k}} {\stateAtPos{q}{i}} \\
    & \qquad \cdot \frac{  \condProb{\forall i{<}j{<}k: \stackAt j > \stackAt i= \stackAt k } {\stateAtPos{q}{i} \land \stepAtPos{k} \land \stateAtPos{r}{k}}}{\condProb{\stepAtPos{i}}{ \stateAtPos{q}{i}}} \bigg)
    \end{split}
    \tag{cond.\ probability} \\
    = \quad & \sum_{k > i{+}1} \condProb{\stepAtPos{k} \land \stateAtPos{r}{k}} {\stateAtPos{q}{i}} \cdot \frac{  \condProb{\forall i{<}j{<}k: \stackAt j > \stackAt i= \stackAt k } {\stateAtPos{q}{i}\land \stateAtPos{r}{k}}}{\condProb{\stepAtPos{i}}{ \stateAtPos{q}{i}}}
    \tag{by \eqref{eq:stepIndep} }\\
    = \quad &\sum_{k > i{+}1} \frac{ \condProb{\stepAtPos{k}} {\stateAtPos{q}{i} \land \stateAtPos{r}{k}}}{\condProb{\stepAtPos{i}}{ \stateAtPos{q}{i}}} \cdot \condProb{\forall i{<}j{<}k: \stackAt j > \stackAt i= \stackAt k \land \stateAtPos{r}{k}} {\stateAtPos{q}{i}}  \tag{rewriting}\\
    = \quad & \frac{ \nonTermProb{r} } {\nonTermProb{q}} \sum_{k > i{+}1} \condProb{ \stateAtPos{r}{k} \land \forall i{<}j{<}k: \stackAt j > \stackAt i=\stackAt k } {\stateAtPos{q}{i} }
    \tag{by \eqref{eq:stepProbIsNonTermProb} and noticing that $\condProb{\stepAtPos{k}} {\stateAtPos{q}{i} \land \stateAtPos{r}{k}} = \condProb{\stepAtPos{k}} {\stateAtPos{r}{k}}$} \\
    = \quad & \frac{ \nonTermProb{r} } {\nonTermProb{q}} \sum_{r', Z} \ppdaTransFun_{\callSubscript}(q, r'Z)\termProbFromTo{r'Z}{r} 
    \tag{by \eqref{eq:returnProbsStackHeights}}
\end{align*}}This concludes the proof of \Cref{thm:transProbs}.
\qed

\subsubsection{The step chain.}

\begin{figure}[t]
    \centering
    \begin{adjustbox}{max height=50mm}
        \begin{tikzpicture}[myautomaton,bend angle=10,node distance=20mm,on grid]
            \node[state,initial] (tau) {$\tau$};
            \node[retState,below left=of tau] (r) {$r$};
            \node[callState,below right=of tau] (c) {$c$};
            
            \draw[->] (tau) -- node[above left] {$\nicefrac 1 3$} (r);
            \draw[->] (tau) -- node {$\nicefrac 2 3$} (c);
            \draw[->] (r) edge[bend right] node[below,align=left] {$Z, \nicefrac 2 3$ \\[-2pt] $\botStack, \nicefrac 2 3$} (c);
            \draw[->] (c) edge[bend right] node[above] {$\nicefrac 1 3, Z$} (r);
            \draw[->] (r) edge[loop below] node[align=left,left=1mm] {$Z, \nicefrac 1 3$ \\[-2pt] $\botStack, \nicefrac 1 3$} (r);
            \draw[->] (c) edge[loop below] node[align=left] {$\nicefrac 2 3, Z$} (c);
        \end{tikzpicture}
        \hspace{12mm}
        \begin{tikzpicture}[myautomaton,bend angle=10,node distance=20mm,on grid]
            \node[state,initial] (taubot) {$\tau\botStack$};
            \node[state,below left=of taubot] (rbot) {$r\botStack$};
            \node[state,below right=of taubot] (cbot) {$c\botStack$};
            \node[state,right=of cbot] (c) {$c$};
            \node[state,right=of taubot] (tau) {$\tau$};
            
            \draw[->] (taubot) edge node[above left] {$\nicefrac 1 3$} (rbot);
            \draw[->] (taubot) edge node[above right] {$\nicefrac 2 3$} (cbot);
            \draw[->] (rbot) edge[loop below] node {$\nicefrac 1 3$} (rbot);
            \draw[->] (rbot) edge[bend right] node[below] {$\nicefrac 2 3$} (cbot);
            \draw[->] (cbot) edge[loop below] node {$\nicefrac 1 3$} (cbot);
            \draw[->] (cbot) edge[bend right] node[above] {$ \nicefrac 1 6$} (rbot);
            \draw[->] (cbot) edge node[below] {$\nicefrac 1 2$} (c);
            \draw[->] (tau) edge node[above right] {$1$} (c);
            \draw[->] (c) edge[loop below] node {$1$} (c);
        \end{tikzpicture}
    \end{adjustbox}
    \caption{Left: Example pVPA $\ppda$ with the following return and diverge probabilities:
        $\termProbFromTo{cZ}{c} = \nicefrac 1 6$,
        $\termProbFromTo{cZ}{r} = \nicefrac{1}{12}$,
        $\termProbFromTo{rZ}{r} = \nicefrac 1 3$,
        $\termProbFromTo{rZ}{c} = \nicefrac 2 3$,
        and $\nonTermProb{c} = \nicefrac 3 4$,
        $\nonTermProb{\tau} = \nicefrac 1 2$,
        $\nonTermProb{r} = 0$.
        In general, these probabilities may be irrational numbers~\cite{DBLP:journals/jacm/EtessamiY09}.
        Right: The step chain $\stepMC{\ppda}$ according to \Cref{def:stepMC}.
        The transition probabilities can be computed using the return and diverge probabilities and \Cref{fig:transProbs}.
    }
    \label{fig:stepMcExampleAdvanced}
\end{figure}

Recall from \eqref{eq:markov_prop} that the stochastic process $\V{}{0},\V{}{1} \ldots$, where $\V{}{i} \in \ppdaStates \cup \ppdaStates\bot$ is the extended  state at the $i$th step, has the Markov property.
Since $\ppdaStates \cup \ppdaStates\bot$ is a finite set, we can now use \Cref{thm:transProbs} to construct the underlying finite Markov chain explicitly.

\begin{definition}[{The Step Chain $\stepMC{\ppda}$}]
    \label{def:stepMC}
    $\stepMC{\ppda}$ is the Markov chain with states
    \[
        M \qeq \Set{q \in \ppdaStatesCall \cup \ppdaStatesInt }{ \nonTermProb{q} > 0 } ~\cup~ \ppdaStates\botStack ~,
    \]
    initial state $\sInitPpda\botStack$, and for all $v, v' \in M$, the probability of transition $v \to v'$ is defined according to \Cref{fig:transProbs}.
    \qeddef
\end{definition}

\Cref{fig:stepMcExampleAdvanced} depicts a non-trivial pVPA and its step chain.
In this example, all return and diverge probabilities are rational.
In general, however, the return and diverge probabilities (\Cref{def:returnDiverge}) are algebraic numbers that are not always rational or even expressible by radicals~\cite{DBLP:journals/jacm/EtessamiY09} (cf.\ \Cref{ex:nonRationalProbs}).
As a consequence, one cannot easily perform numerical computations on the step chain.
However, the probabilities can be encoded implicitly as the unique solution of an \emph{existential theory of the reals} (ETR) formula, i.e., an existentially quantified FO-formula over $(\mathbb{R}, +, \cdot, \leq)$~\cite{esparzaMCPPDA-lics}.
Since the ETR is decidable in $\PSPACE$, many questions about the step chain are in $\PSPACE$ as well.
We will make use of this in \Cref{thm:resStPaDVPA} below.

The property of $\stepMC{\ppda}$ that is most relevant to us is given by the following \Cref{thm:stepChainSound}.
For $\run \in \Runs_{\ppda}$, we let $\footprint{\run} = \V{}{0}(\run)\V{}{1}(\run)\ldots \in (\ppdaStates \cup \ppdaStates\botStack)^\omega$ (note that this is a slightly different ``footprint'' than the one introduced in \Cref{sec:stair-parity}).

\begin{restatable}[Soundness of $\stepMC{\ppda}$]{lemma}{stepChainSound}
    \label{thm:stepChainSound}
    Let $\ppda$ be a pVPA with step chain $\stepMC{\ppda}$.
    Let $M$ be the states of the step chain and let $R \subseteq M^\omega$ be measurable.
    Then
    \[
        \PP_{\ppda} ( \Set{\run \in \Runs_{\ppda}}{\footprint{\run} \in R} ) \qeq \PP_{\stepMC{\ppda}}(R)
    \]
    where $\PP_{\ppda}$ and $\PP_{\stepMC{\ppda}}$ are the probability measures associated with $\ppda$ and $\stepMC{\ppda}$, respectively.
\end{restatable}
\begin{proof}
    The formal proof requires some basic notions from measure theory.
    In fact, \Cref{thm:stepChainSound} is actually an instance of the following more general statement:
    
    \begin{enumerate}
        \item[($\star$)] \emph{Let $(X,\mathcal A, \mu)$ and $(Y,\mathcal B, \nu)$ be probability spaces such that $\mathcal{B} = \sigma(\mathcal G)$, where $\mathcal{G} \subseteq 2^Y$, i.e., $\mathcal{B}$ is the $\sigma$-algebra generated by the sets in $\mathcal{G}$.
        Assume furthermore that $\mathcal G$ is a $\pi$-system, i.e., $\mathcal G$ is non-empty and closed under finite intersections.
        Let $f \colon X \to Y$ be such that for all $G \in \mathcal G$, $f^{-1}(G) \in \mathcal A$ and $\mu(f^{-1}(G)) = \nu(G)$.
        Then $f$ is a measurable function and the pushforward measure $\nu' = \mu \circ f^{-1}$ coincides with $\nu$.}
    \end{enumerate}

    We now explain how to prove ($\star$) using fundamental measure theory.
    The fact that $f$ is measurable follows because the inverse image $f^{-1}$ preserves set operations (see, e.g., \cite[below~Definition~1.5.1]{ash2000probability}).
    For the claim that $\nu' = \nu$ it suffices to note that by assumption we have for all $G \in \mathcal{G}$ that $\nu'(G) = \nu(G)$, and since $\mathcal G$ is a $\pi$-system, an application of the $\pi$-$\lambda$ theorem (see, e.g.,~\cite[Proposition~2.10]{panangaden2009labelled}) implies that $\nu' = \nu$.
    
    We instantiate ($\star$) as follows to prove \Cref{thm:stepChainSound}:
    The probability spaces are the ones associated with the measures $\PP_{\ppda}$ and $\PP_{\stepMC{\ppda}}$.
    In particular, the $\sigma$-algebra on which $\PP_{\stepMC{\ppda}}$ is defined is the one generated by the cylinder sets $\mathcal C = \{w.M^\omega \mid w \in M^*\} \subseteq 2^{M^\omega}$.
    It is easy to see that $\mathcal G = \mathcal{C} \cup \{\emptyset\}$ is a $\pi$-system, and $\sigma(\mathcal C) = \sigma(\mathcal G)$.
    We define $f \colon (\ppdaStates \times \abStackNoBot^*.\botStack)^\omega \to M^\omega, f(\run) = \footprint{\run}$, i.e., $f$ projects a run from $\ppda$ to its footprint of steps (which is a run in $\stepMC{\ppda}$).
    To apply $(\star)$ it remains to show that for all cylinder sets $w.M^\omega$, $w \in M^*$, we have that
    (i) $f^{-1}(w.M^\omega)$ is measurable, and
    (ii) $\PP_{\ppda}(f^{-1}(w.M^\omega)) = \PP_{\stepMC{\ppda}}(w.M^\omega)$.
    For (i) notice that $f^{-1}(w.M^\omega) = \bigwedge_{i=0}^{|w|-1}(\V{}{i} = w(i))$ is indeed measurable because it is a finite intersection of measurable events by \Cref{thm:VwellDef}; recall that $\V{}{i}$ denotes the (extended) state at the $i$th step.
    (ii) is trivial in the case where $|w| = 0$, so we let $|w| = n+1$, $n \geq 0$, and exploit the properties of the step chain $\stepMC{\ppda}$.
    If $w(0) \neq q_0\botStack$ (the initial state of $\stepMC{\ppda}$) then $\PP_{\ppda}(\V{}{0} = w(0)) = \PP_{\stepMC{\ppda}}(w.M^\omega) = 0$.
    Otherwise $w(0) = q_0\botStack$.
    In this case, if $n = 0$ (i.e., $|w| = 1$), then  $\PP_{\ppda}(\V{}{0} = w(0)) = \PP_{\stepMC{\ppda}}(w.M^\omega) = 1$.
    Else, if $n > 0$ ($|w| > 1$), by the Markov property from equation \eqref{eq:markov_prop}, we have
    \begin{align*}
        &\PP_{\ppda}(\V{}{n} = w(n) \land  \ldots \land \V{}{0} = w(0)) \\
        \qeq &  \PP_{\ppda}(\V{}{n} = w(n) \mid \V{}{n-1} = w(n-1) ) \,\cdot \,\ldots\, \cdot \,\PP_{\ppda}(\V{}{1} = w(1) \mid \V{}{0} = w(0)) \\
        \qeq & P(w(n-1), w(n)) \, \cdot \, \ldots \, \cdot \, P(w(0), w(1)) \tag{By \Cref{thm:transProbs}}\\
        \qeq & \PP_{\stepMC{\ppda}}(w.M^\omega)
    \end{align*}
    where $P$ is the transition probability function in the Markov chain $\stepMC{\ppda}$ and the last equality holds by definition of the probability measure $\PP_{\stepMC{\ppda}}$.
\end{proof}

\subsection{Main Result of \texorpdfstring{\Cref{sec:modelCheckingDVPA}}{Section 4}}
\label{sec:proofMainRes}
The following is the main result of \Cref{sec:modelCheckingDVPA}:

\begin{restatable}{theorem}{restateMainResDet}
    \label{thm:resStPaDVPA}
    Let $\ppda$ be a pVPA and $\dvpa$ be a stair-parity DVPA, both over the same pushdown alphabet $\ab$.
    Then for all $\theta \in [0,1] \cap \mathbb{Q}$, the following is decidable in $\PSPACE$:
    \[
        \PP(\{\run \in \Runs_\ppda \mid \labeling(\run) \in \mathcal{L}(\dvpa)\}) \geq_{?} \theta
        ~.
    \]
\end{restatable}

The rest of \Cref{sec:proofMainRes} is devoted to the proof of \Cref{thm:resStPaDVPA}.
We first provide a brief overview.
The first step is to construct the product $\automProd{\ppda}{\dvpa}$ according to \Cref{def:prod}.
By \Cref{thm:product} we need to compute the stair-parity acceptance probability of $\automProd{\ppda}{\dvpa}$.
\Cref{thm:stepChainSound} reduces this to computing a usual parity acceptance probability in the step chain $\stepMC{\automProd{\ppda}{\dvpa}}$.
This can be achieved through finding the bottom strongly connected components (BSCC) of $\stepMC{\automProd{\ppda}{\dvpa}}$, classifying them as \emph{good} (the minimum priority of a BSCC state is even) or otherwise as \emph{bad}, and running a standard reachability analysis w.r.t.\ the good BSCCs (see \Cref{fig:productAndStepChain} for an example).
The remaining technical difficulty is that the transition probabilities of $\stepMC{\automProd{\ppda}{\dvpa}}$ are not rational in general.
We handle this using the fact that these probabilities are expressible in the ETR~\cite{esparzaMCPPDA-lics}.

\newcommand{\termProbFromToVar}[2]{\langle #1 {\downarrow} #2\rangle}
\newcommand{\prodAbbr}{\tilde{\ppda}}

We now present the formal proof.
We use the following result about return probabilities of pPDA, which is originally due to~\cite{esparzaMCPPDA-lics}:
\begin{lemma}[{as stated in~\cite[Theorem 2]{brazdilSurvey}}]
    \label{thm:retProbsFormula}
    The return probabilities $\termProbFromTo{pZ}{q}$ are expressible in ETR.
    More specifically, there exists an FO-formula $\Phi$ over $(\mathbb{R}, +, \cdot, \leq)$ which uses just existential quantifiers and free variables $\termProbFromToVar{pZ}{q}_{p,q \in\ppdaStates, Z\in\abStack}$ such that $\Phi$ becomes a true FO-sentence iff each free variable $\termProbFromToVar{pZ}{q}$ is substituted by the actual return probability $\termProbFromTo{pZ}{q}$.
    Moreover, $\Phi$ can be effectively constructed in polynomial space.
\end{lemma}

\begin{lemma}
    \label{thm:transProbETR}
    The next-step probabilities (i.e., the transition probabilities of the step chain) in \Cref{fig:transProbs} are expressible in ETR.
\end{lemma}
\begin{proof}
    With \Cref{thm:retProbsFormula} it suffices to note that ETR expressible numbers are closed under addition, multiplication and division.
    Let $x,y \in \mathbb{R}$ be expressed by ETR formulae $\Phi(x)$ and $\Phi'(y)$, respectively.
    Then the formula $\Phi''(z) \coloneqq \exists x,y \colon \Phi(x) \land \Phi'(y) \land z = x+y$ where $z$ is a fresh variable expresses the sum of $x$ and $y$, and similar for multiplication.
    For division, we have that $\Phi''(z) \coloneqq \exists x,y \colon \Phi(x) \land \Phi'(y) \land z \cdot y  = x$ expresses $x/y$ provided that $y \neq 0$ (if $y = 0$ then $\Phi''$ does not express a unique real number as it is then either unsatisfiable or trivial).
\end{proof}

    \begin{table}[t]
    \caption{
        The underlying graph of the step chain.
        The condition in each cell is true iff the corresponding transition probability in \Cref{fig:transProbs} is non-zero.
    }
    \label{fig:constructG}
    \centering
    {\renewcommand{\arraystretch}{1.1}
        \begin{adjustbox}{max width=\textwidth}
            \begin{tabular}{l  c  c  c  c }
                \toprule
                & $q \to r$& $q \bot \to r $ & $q \bot \to r \bot$ & $q \to r \bot$\\
                \midrule
                $q \in \ppdaStatesCall ~$
                & \makecell[c]{$ \nonTermProb{r} > 0 \lland ($ \\ $(\exists r',Z \colon \ppdaTransFun_{\callSubscript}(q,r'Z) > 0 \land \termProbFromTo{r'Z}{r}) > 0) $ \\  $\llor \exists Z \colon \ppdaTransFun_{\callSubscript}(q,rZ) > 0) $} \qquad
                & \makecell[c]{$\exists Z \colon \ppdaTransFun_{\callSubscript}(q,rZ) > 0 $ \\ $\lland \nonTermProb{r} > 0$} \qquad
                & \makecell[c]{$\exists r', Z \colon \ppdaTransFun_{\callSubscript}(q,rZ) > 0 $ \\ $\lland \termProbFromTo{r'Z}{r} > 0$} 
                & false\\
                \midrule
                $q \in \ppdaStatesInt$
                & $\nonTermProb{r} > 0 \lland \ppdaTransFun_{\intSubscript}(q, r) >0 $
                & false
                & $\ppdaTransFun_{\intSubscript}(q,r) > 0$
                & false\\
                \midrule
                $q \in \ppdaStatesRet$
                & false
                & false
                & $\ppdaTransFun_{\retSubscript}(q\botStack, r) > 0$
                & false\\\bottomrule
            \end{tabular}
        \end{adjustbox}
    }
\end{table}

We now describe our $\PSPACE$ algorithm to prove \Cref{thm:resStPaDVPA}.

\textbf{Step 1.}
We first construct the product pVPA $\prodAbbr = \automProd{\ppda}{\dvpa}$ with priority function $\prioFun' \colon \ppdaStates \to \Nats$ where $\ppdaStates$ are the states of $\prodAbbr$ as in \Cref{def:prod}.
By \Cref{thm:product} it suffices to compute the probability
\begin{align}
\label{eq:probInProduct}
\PP(\Set{\pi \in \Runs_{\prodAbbr}}{\stepRestr{\run} \in \accSet{\prioFun'}})
\end{align}
that the footprint of a run in the product $\prodAbbr$ satisfies the parity condition induced by $\prioFun'$.
$\prodAbbr$ can be constructed in polynomial time.

\textbf{Step 2.}
We express \eqref{eq:probInProduct} using the step chain $\stepMC{\prodAbbr}$.
Let $M \subseteq \ppdaStates \cup \ppdaStates\bot$ be the states of the step chain $\stepMC{\prodAbbr}$.
Let $\prioFun'' \colon M \to \Nats$ be the extension of $\prioFun'$ to the states of $M$ via $\prioFun''(q) = \prioFun''(q\bot) = \prioFun'(q)$ for all $q \in \ppdaStates$.
That is, $\prioFun''$ induces a parity acceptance set $\accSet{\prioFun''} \subseteq M^\omega$ which is $\omega$-regular and thus measurable.
Let $\run \in \Runs_{\prodAbbr}$.
Clearly, $\stepRestr{\run} \in \accSet{\prioFun'}$ iff $\footprint{\run} \in \accSet{\prioFun''}$.
Thus by \Cref{thm:stepChainSound}, \eqref{eq:probInProduct} is equal to
\[
\PP(\Set{\pi \in \Runs_{\prodAbbr}}{\stepRestr{\run} \in \accSet{\prioFun'}})
\qeq
\PP(\accSet{\prioFun''})
\]
where the right hand side is the probability that a run of the step chain $\stepMC{\prodAbbr}$ satisfies the parity condition induced by $\prioFun''$.
We have thus reduced the problem of computing a stair-parity acceptance probability in the product pPDA $\prodAbbr$ to computing a standard parity acceptance probability in the finite Markov chain $\stepMC{\tilde{\ppda}}$.

The rest of the proof uses standard techniques and is similar to the proof of~\cite[Theorem 5.15]{esparzaMCPPDA-lics}.
The main technical difficulty is that the transition probabilities of $\stepMC{\prodAbbr}$ cannot be written in an explicit numerical form since they are in general algebraic numbers.

\textbf{Step 3.}
We construct the \emph{underlying graph} $G_{\tilde{\ppda}}$ of the step chain $\stepMC{\tilde{\ppda}}$, i.e., we determine the set $M$ of states and include a directed edge between states $v, v' \in M$ iff the probability of transition $v \to v'$ is positive.
\Cref{fig:constructG} gives sufficient and necessary conditions for this in all cases (\Cref{fig:constructG} can be seen as the ``qualitative'' version of \Cref{fig:transProbs}).
The conditions in \Cref{fig:constructG} (as well as constructing the state space of $\stepMC{\tilde{\ppda}}$) require checking if $\nonTermProb{q} > 0$ for states $q \in \ppdaStates$ of $\tilde{\ppda}$.
The latter is equivalent to checking if $\sum_{r \in \ppdaStates} \termProbFromTo{qZ}{r} < 1$, which is reducible to ETR by \Cref{thm:retProbsFormula} and hence decidable in $\PSPACE$.

\textbf{Step 4.}
We determine the bottom strongly connected components (BSCC) of $G_{\tilde{\ppda}}$ from the previous step by a standard (efficient) graph analysis.
We mark the BSCCs $B \subseteq M$ such that $\min_{v \in B} \prioFun''(v)$ is even as ``good'', the others as ``bad''.
It is well-known that in the finite Markov chain $\stepMC{\prodAbbr}$ it holds that
\[
    \condProb{\run \in \accSet{\prioFun''}}{\run \text{ reaches a good BSCC}} =1
\]
due to the Long-Run Theorem of finite Markov chains: Each state of a BSCC is visited infinitely often almost-surely, provided that this BSCC is reached at all.
Moreover, if a run $\run$ reaches a bad BSCC, than the probability to satisfy the parity condition is zero and thus
\begin{align*}
    \PP(\run \in \accSet{\prioFun''})
    ~=~
    & \condProb{\run \in \accSet{\prioFun''}}{\run \text{ reaches a good BSCC}} \cdot \PP(\run \text{ reaches a good BSCC}) \\
    ~=~
    & \PP(\run \text{ reaches a good BSCC})
\end{align*}
Thus it only remains to compute the probability to reach a good BSCC in $\stepMC{\prodAbbr}$.

\textbf{Step 5.}
We use the previous step to classify the states $M$ of the step chain $\stepMC{\prodAbbr}$ into three categories: $M_{=0}$ contains all states from which no good BSCC is reachable in the graph $G_{\tilde{\ppda}}$, $M_{=1}$ contains all good BSCCs, and $M_{?}$ contains all other states.
We recall that the probabilities to reach $M_{=1}$ are the \emph{unique} solution of the following linear equation system (see, e.g., \cite[Ch.\ 10]{MCbible}):
\begin{align*}
    x_v &= 0 & \text{ if } v \in M_{=0} \\
    \land \quad x_v &= 1 & \text{ if } v \in M_{=1} \\
    \land \quad x_v &= \sum_{v \xrightarrow{p} v'} px_{v'}  & \text{ if } v \in M_{?}~.
\end{align*}
We can treat the vectors of probabilities $\vec{p}$ and the variables $\vec{x}$ in this equation system as free variables of an ETR formula $\mathcal{R}(\vec{p},\vec{x})$.
By \Cref{thm:transProbETR}, there is an ETR formula $\Phi(\vec{p})$ expressing $\vec{p}$.
The ETR formula
\[
    \exists \vec{p}, \vec{x} \colon \quad \Phi(\vec{p}) \qland \mathcal{R}(\vec{p},\vec{x}) \qland x_{v_0} \geq \theta
\]
is thus true iff the probability to reach a good BSCC from initial state $v_0 \in M$ is at least $\theta$.
Truth of this formula can be decided in $\PSPACE$, which concludes the proof of \Cref{thm:resStPaDVPA}.
\qed

\begin{figure}[t]
    \centering
    \begin{adjustbox}{max height=100mm}
        \begin{tikzpicture}[myautomaton,bend angle=10,node distance=18mm and 30mm,on grid]
            \node[state, initial,inner sep=1pt,minimum size=9mm] (taus0) {$\tau s_0$};
            \node[retState, right=of taus0,inner sep=3pt] (rs1) {$r s_1$};
            \node[callState, below=of taus0,minimum size=7.5mm] (cs1) {$c s_1$};
            \node[callState, below=of cs1,minimum size=7.5mm] (cs0) {$c s_0$};
            \node[retState, right=of cs0,inner sep=3pt] (rs0) {$r s_0$};
            
            \draw[->] (taus0) -- node {$\nicefrac 1 3$} (rs1);
            \draw[->] (taus0) -- node[left] {$\nicefrac 2 3$} (cs1);
            \draw[->] (rs1) edge[loop right] node[below=2mm] {$\langle *, *\rangle, \nicefrac 1 3$} (rs1);
            \draw[->] (rs1) -- node[sloped] {$\langle *, *\rangle, \nicefrac 2 3$} (cs1);
            \draw[->] (cs1) -- node[left] {$\nicefrac 2 3, \langle Z, Z\rangle$} (cs0);
            \draw[->] (cs1) edge[bend left] node[sloped] {$\nicefrac 1 3, \langle Z, Z\rangle$} (rs0);
            \draw[->] (cs0) edge[loop left] node[below=2mm,xshift=-3mm] {$\nicefrac 2 3, \langle Z, Z\rangle$} (cs0);
            \draw[->] (cs0) -- node[below] {$\nicefrac 1 3, \langle Z, Z\rangle$} (rs0);
            \draw[->] (rs0) edge[bend left] node[below,sloped,xshift=-2mm] {$\langle *, *\rangle, \nicefrac 2 3$} (cs1);
            \draw[->] (rs0) -- node[right] {$\langle *, *\rangle, \nicefrac 1 3$} (rs1);
        \end{tikzpicture}
        \hspace{1cm}
        \begin{tikzpicture}[myautomaton,bend angle=10,node distance=18mm and 30mm,on grid]
            \tikzstyle{myState} = [state,inner sep=1pt,minimum size=9mm]
            
            \node[myState, initial] (taus0bot) {$\tau s_0\bot$};
            \node[myState, right=of taus0bot] (rs1bot) {$r s_1\bot$};
            \node[myState, below=of taus0bot] (cs1bot) {$c s_1\bot$};
            \node[myState, right=of cs1bot] (cs1) {$c s_1$};
            \node[myState, below=of cs1] (cs0) {$c s_0$};
            
            \draw[dashed] (2.25,-1) -- (3.75,-1) -- (3.75, -4.25) -- (2.25,-4.25) -- cycle;
            
            \draw[->] (taus0bot) -- node {$\nicefrac 1 3$} (rs1bot);
            \draw[->] (taus0bot) -- node[left] {$\nicefrac 2 3$} (cs1bot);
            \draw[->] (rs1bot) edge[loop right, looseness=5] node[below=2pt] {$\nicefrac 1 3$} (rs1bot);
            \draw[->] (rs1bot) edge[bend left] node[below] {$\nicefrac 2 3$} (cs1bot);
            \draw[->] (cs1bot) edge[loop left, looseness=5] node[below=2pt] {$\nicefrac 1 3$} (cs1bot);
            \draw[->] (cs1bot) edge[bend left] node[above,near start] {$\nicefrac 1 6$} (rs1bot);
            \draw[->] (cs1bot) -- node[above, near start] {$\nicefrac 1 2$} (cs0);
            \draw[->] (cs1) edge[loop right] node {$\nicefrac 1 3$} (cs1);
            \draw[->] (cs1) edge[bend left] node {$\nicefrac 2 3$} (cs0);
            \draw[->] (cs0) edge[loop right] node {$\nicefrac 2 3$} (cs0);
            \draw[->] (cs0) edge[bend left] node {$\nicefrac 1 3$} (cs1);
        \end{tikzpicture}
    \end{adjustbox}
    \caption{
        Left: The product of the pVPA from \Cref{fig:stepMcExampleAdvanced} (left) and the DVPA from \Cref{fig:stepsExample} on page \pageref{fig:stepsExample}.
        Right: Its step chain according to \Cref{def:stepMC}.
        The dashed region is the only BSCC.
        It violates the parity condition $\prioFun(s_0) = 1$ and $\prioFun(s_1) = 2$ inherited from the DVPA (see \Cref{ex:repbdd} on page \pageref{ex:repbdd}) since every run reaching the BSCC visits $cs_0$ infinitely often with probability $1$.
        Only reachable states are depicted.
        $*$ is a placeholder that stands for an arbitrary stack symbol.
    }
    \label{fig:productAndStepChain}
\end{figure}

\subsection{Implications for Probabilistic One-counter Automata}

A probabilistic visibly \emph{one-counter automaton} (pVOC) is the special case of a pVPA with unary stack alphabet, i.e., $|\abStackNoBot| = 1$.
For example, the pVPA in \Cref{fig:stepMcExampleAdvanced} (left) is a pVOC.
For many problems, better complexity bounds are known for pVOC than for the general case.
In particular, $\nonTermProb{p} >_{?} 0$, i.e., the question whether a pVOC started in state $p$ never reaches counter value (or stack height) zero with positive probability, can be decided in $\PTIME$~\cite[Theorem~4]{brazdilSurvey}.
We can exploit this to improve \Cref{thm:resStPaDVPA} in the pVOC case:

\begin{corollary}
    \label{thm:resStPaDVPAqual}
    Let $\ppda$ be a pVOC and $\dvpa$ be a stair-parity DVPA over pushdown alphabet $\ab$.
    The problem $\PP(\{\run \in \Runs_\ppda \mid \labeling(\run) \in \mathcal{L}(\dvpa)\}) =_? 1$ is decidable in $\PTIME$.
\end{corollary}
\begin{proof}
    The key observation is that, since we can efficiently decide $\nonTermProb{p} >_{?} 0$, we can efficiently (in polynomial time) construct the underlying graph $G_{\automProd{\ppda}{\dvpa}}$ of the step chain of $\automProd{\ppda}{\dvpa}$ (as in the proof of \Cref{thm:resStPaDVPA}), and then apply polynomial-time graph analysis algorithms to check if only good BSCCs are reachable in $G_{\automProd{\ppda}{\dvpa}}$.
\end{proof}

\Cref{thm:resStPaDVPAqual} implies that there exist efficient algorithms for many properties of pVOC-expressible random walks on $\Nats$.
In fact, almost-sure satisfaction of each \emph{fixed} visibly-pushdown property can be decided in $\PTIME$.
For instance, using the DVPA from \Cref{fig:stepsExample} we can decide if a random walk is repeatedly bounded with probability $1$.


\section{Model Checking against Büchi VPA and CaRet}
\label{sec:modelCheckingNVPACaRet}

With \Cref{thm:StPaDVPA,thm:resStPaDVPA} it follows immediately that quantitative model checking of pVPA against non-deterministic Büchi VPA is decidable in $\EXPSPACE$.
We can improve the complexity in the qualitative case:
\begin{restatable}{theorem}{thmresNVPAQual}
    \label{thm:resNVPAQual}
    Let $\ppda$ be a pVPA and $\npda$ be a (non-deterministic) Büchi VPA over the same pushdown alphabet.
    The problem $\PP(\{\run \in \Runs_\ppda \mid \labeling(\run) \in \mathcal{L}(\npda)\}) =_? 1$ is $\EXPTIME$-complete.
\end{restatable}
\begin{proof}
    The lower bound is due to \cite[Theorem~8]{DBLP:conf/tacas/EtessamiY05} and already holds for non-pushdown Büchi automata.
    We now describe an $\EXPTIME$ decision procedure:
    \begin{itemize}
        \item We first determinize $\npda$ using \Cref{thm:StPaDVPA} which is possible in time exponential in $|\npda|$.
        Let $\dvpa$ be the resulting stair-parity DVPA and consider the product $\automProd{\ppda}{\dvpa}$ (\Cref{def:prod}).
        Note that the product can be constructed in polynomial time in $|\dvpa|$ and $|\ppda|$, and thus in exponential time in the overall size of the input.
        \item The crucial observation for the next step is that $\nonTermProb{q} = \nonTermProb{(q,s)}$ for all states $(q,s)$ of $\automProd{\ppda}{\dvpa}$.
        This holds because by definition of the product, $\dvpa$ merely \emph{observes} the runs of $\ppda$, and thus the diverge probabilities of $\automProd{\ppda}{\dvpa}$ and $\ppda$ are essentially the same.
        We compute the set $Div_{\ppda} = \{q \mid \nonTermProb{q} > 0\} \subseteq  Q$, where $Q$ are the states of $\ppda$, in \emph{exponential time in $|\ppda|$} using a $\PSPACE$ decision procedure for the ETR~\cite{esparzaMCPPDA-lics}.
        Note that computing $Div_{\automProd{\ppda}{\dvpa}} = \{(q,s) \mid \nonTermProb{(q,s)} > 0\}$ directly would take \emph{doubly-exponential} time in $|\npda|$; the proposed ``optimization'' is thus essential for obtaining the $\EXPTIME$ upper bound.
        \item We now determine the set of triples $Ret_{\automProd{\ppda}{\dvpa}} = \{(q,Z,p) \mid \termProbFromTo{qZ}{p} > 0 \}$ in $\automProd{\ppda}{\dvpa}$.
        Unlike the diverge probabilities, this set can be computed in \emph{polynomial} time in the size of $\automProd{\ppda}{\dvpa}$ (hence exponential in the size of the input) because we may disregard the exact transition probabilities and conduct a standard reachability analysis in a non-deterministic pushdown automaton~\cite{DBLP:reference/mc/AlurBE18}, also see \cite[p.~136]{brazdilSurvey}.
        \item The next step is to construct the underlying \emph{graph} $G_{\automProd{\ppda}{\dvpa}}$ of the step chain $\stepMC{\automProd{\ppda}{\dvpa}}$, i.e., the directed graph that has the same vertices as $\stepMC{\automProd{\ppda}{\dvpa}}$ and includes an edge $(u,v)$ iff the 1-step transition probability from $u$ to $v$ is positive in the Markov chain.
        This can be done in polynomial time in $|\ppda \times \dvpa|$ using the sets $Div_{\ppda}$ and $Ret_{\automProd{\ppda}{\dvpa}}$ defined above and \Cref{fig:constructG}.
        \item The final step is, as in \Cref{thm:resStPaDVPA}, to determine the BSCCs of $G_{\automProd{\ppda}{\dvpa}}$, classify them as good or bad according to whether they satisfy the (standard) parity condition inherited from $\dvpa$, and then check if there is a bad BSCC reachable from the initial state.
        All these steps can be done in polynomial time in $|\automProd{\ppda}{\dvpa}|$.
        \qedhere        
    \end{itemize}
\end{proof}    

In the above result, membership in $\EXPTIME$ relies on the fact that one can construct the underlying \emph{graph} of a step chain $\stepMC{\automProd{\ppda}{\dvpa}}$ in time exponential in the size of $\ppda$ but \emph{polynomial} in the size of $\dvpa$.
$\EXPTIME$-hardness follows from~\cite[Theorem~8]{DBLP:conf/tacas/EtessamiY05}.
In fact, qualitative model checking of pPDA against \emph{non-pushdown} Büchi automata is also $\EXPTIME$-complete~\cite{DBLP:conf/tacas/EtessamiY05}.
With \Cref{thm:caretToNBVPA,thm:StPaDVPA,thm:resStPaDVPA,thm:resNVPAQual} we immediately obtain the following complexity results for CaRet model checking:
\begin{theorem}
    \label{thm:resCaRet}
    The quantitative and qualitative probabilistic CaRet model checking problems (\Cref{def:caretDecisionProblems}) are decidable in $\TwoEXPSPACE$ and $\TwoEXPTIME$, respectively.
\end{theorem}
Both problems are known to be $\EXPTIME$-hard~\cite{DBLP:conf/qest/YannakakisE05}.


\section{Conclusion}
\label{sec:conclusion}

We have presented the first decidability result for model checking probabilistic pushdown automata---an operational model of recursive discrete probabilistic programs---against CaRet, or more generally, against the class of $\omega$-VPL.
We heavily rely on the determinization procedure from~\cite[Theorem~1]{vpGames} and the notion of a step chain used in previous works~\cite{esparzaMCPPDA-lics, esparzaMCPPDA}.
These two constructions turn out to be a natural match.

We conjecture that the upper bounds from \Cref{thm:resCaRet} are not tight due to the exponential blow up incurred by applying the VPA determinization from~\cite[Theorem~1]{vpGames}.
Future work is thus to investigate whether the doubly-exponential complexity can be lowered to singly-exponential, e.g., by generalizing the automata-less algorithm from~\cite{DBLP:conf/qest/YannakakisE05}.
Another open question is whether existing results~\cite{DBLP:journals/jacm/EtessamiY09,DBLP:journals/siamcomp/EsparzaKL10,DBLP:journals/jacm/StewartEY15} for \emph{approximately} computing the probabilities $\termProbFromTo{qZ}{r}$ can be used for approximate quantitative CaRet and $\omega$-VPL model checking.
We also plan to extend our recent work on \emph{certificates}~\cite{DBLP:conf/tacas/WinklerK23,DBLP:conf/lics/WinklerK23} to temporal and other logical properties.
Such certificates can be approximate as well.

Other future work includes exploring to what extent algorithms for probabilistic CTL can be generalized to the branching-time variant of CaReT~\cite{DBLP:conf/spin/GutsfeldMN18},
considering more expressive logics such as visibly LTL~\cite{DBLP:journals/jar/BozzelliS18} or OPTL~\cite{DBLP:journals/tcs/ChiariMP20},
and studying the interplay of \emph{conditioning} and recursion~\cite{DBLP:conf/starai/StuhlmullerG12} through the lens of pPDA.

\paragraph{Acknowledgement.}
The authors thank Christof Löding for his pointer to stair-parity VPA, and the anonymous reviewers for their useful suggestions and feedback.
We also thank Darion Haase for helpful discussions regarding the proof of \Cref{thm:stepChainSound}.

%
%
\bibliographystyle{alphaurl}
\bibliography{references-lmcs22}

\newcommand{\etalchar}[1]{$^{#1}$}
\begin{thebibliography}{vdMPYW18}

\bibitem[AAB{\etalchar{+}}07]{DBLP:conf/lics/AlurABEIL07}
Rajeev Alur, Marcelo Arenas, Pablo Barcel{\'{o}}, Kousha Etessami, Neil
  Immerman, and Leonid Libkin.
\newblock {First-Order and Temporal Logics for Nested Words}.
\newblock In {\em {LICS}}, pages 151--160. {IEEE} Computer Society, 2007.
\newblock \href {https://doi.org/10.1109/LICS.2007.19}
  {\path{doi:10.1109/LICS.2007.19}}.

\bibitem[ABE18]{DBLP:reference/mc/AlurBE18}
Rajeev Alur, Ahmed Bouajjani, and Javier Esparza.
\newblock {Model Checking Procedural Programs}.
\newblock In Edmund~M. Clarke, Thomas~A. Henzinger, Helmut Veith, and Roderick
  Bloem, editors, {\em Handbook of Model Checking}, pages 541--572. Springer,
  2018.
\newblock \href {https://doi.org/10.1007/978-3-319-10575-8\_17}
  {\path{doi:10.1007/978-3-319-10575-8\_17}}.

\bibitem[ADD00]{ash2000probability}
Robert~B Ash and Catherine~A Dol{\'e}ans-Dade.
\newblock {\em {Probability and Measure Theory}}.
\newblock Academic press, 2000.

\bibitem[AEM04]{caret}
Rajeev Alur, Kousha Etessami, and P.~Madhusudan.
\newblock {A Temporal Logic of Nested Calls and Returns}.
\newblock In {\em {TACAS}}, volume 2988 of {\em Lecture Notes in Computer
  Science}, pages 467--481. Springer, 2004.
\newblock \href {https://doi.org/10.1007/978-3-540-24730-2\_35}
  {\path{doi:10.1007/978-3-540-24730-2\_35}}.

\bibitem[AK15]{axelrod2015branching}
David Axelrod and Marek Kimmel.
\newblock {\em {Branching Processes in Biology}}.
\newblock Springer-Verlag, 2015.
\newblock \href {https://doi.org/10.1007/978-1-4939-1559-0}
  {\path{doi:10.1007/978-1-4939-1559-0}}.

\bibitem[AM04]{vpl}
Rajeev Alur and P.~Madhusudan.
\newblock Visibly pushdown languages.
\newblock In {\em {STOC}}, pages 202--211. {ACM}, 2004.
\newblock \href {https://doi.org/10.1145/1007352.1007390}
  {\path{doi:10.1145/1007352.1007390}}.

\bibitem[AP09]{DBLP:journals/scp/AudebaudP09}
Philippe Audebaud and Christine Paulin{-}Mohring.
\newblock {Proofs of randomized algorithms in {C}oq}.
\newblock {\em Sci. Comput. Program.}, 74(8):568--589, 2009.
\newblock \href {https://doi.org/10.1016/j.scico.2007.09.002}
  {\path{doi:10.1016/j.scico.2007.09.002}}.

\bibitem[BEKK13]{brazdilSurvey}
Tom{\'{a}}s Br{\'{a}}zdil, Javier Esparza, Stefan Kiefer, and Anton{\'{\i}}n
  Kucera.
\newblock Analyzing probabilistic pushdown automata.
\newblock {\em Formal Methods Syst. Des.}, 43(2):124--163, 2013.
\newblock \href {https://doi.org/10.1007/s10703-012-0166-0}
  {\path{doi:10.1007/s10703-012-0166-0}}.

\bibitem[BK08]{MCbible}
Christel Baier and Joost-Pieter Katoen.
\newblock {\em {Principles of Model Checking}}.
\newblock {MIT} Press, 2008.

\bibitem[BKOB13]{DBLP:journals/toplas/BartheKOB13}
Gilles Barthe, Boris K{\"{o}}pf, Federico Olmedo, and Santiago~Zanella
  B{\'{e}}guelin.
\newblock {Probabilistic Relational Reasoning for Differential Privacy}.
\newblock {\em {ACM} Trans. Program. Lang. Syst.}, 35(3):9:1--9:49, 2013.
\newblock \href {https://doi.org/10.1145/2492061} {\path{doi:10.1145/2492061}}.

\bibitem[BKS05]{BKS05}
Tom{\'{a}}s Br{\'{a}}zdil, Anton{\'{\i}}n Kucera, and Oldrich Strazovsk{\'{y}}.
\newblock {On the Decidability of Temporal Properties of Probabilistic Pushdown
  Automata}.
\newblock In {\em {STACS}}, volume 3404 of {\em Lecture Notes in Computer
  Science}, pages 145--157. Springer, 2005.
\newblock \href {https://doi.org/10.1007/978-3-540-31856-9\_12}
  {\path{doi:10.1007/978-3-540-31856-9\_12}}.

\bibitem[BS18]{DBLP:journals/jar/BozzelliS18}
Laura Bozzelli and C{\'{e}}sar S{\'{a}}nchez.
\newblock {Visibly Linear Temporal Logic}.
\newblock {\em J. Autom. Reason.}, 60(2):177--220, 2018.
\newblock \href {https://doi.org/10.1007/s10817-017-9410-z}
  {\path{doi:10.1007/s10817-017-9410-z}}.

\bibitem[CIRW11]{Cassini2008}
Lorenzo Casini, Phyllis~McKay Illari, Federica Russo, and Jon Williamson.
\newblock {Models for prediction, explanation and control: recursive Bayesian
  networks}.
\newblock {\em THEORIA. Revista de Teor{\'\i}a, Historia y Fundamentos de la
  Ciencia}, 26(1):5--33, 2011.

\bibitem[CMP20]{DBLP:journals/tcs/ChiariMP20}
Michele Chiari, Dino Mandrioli, and Matteo Pradella.
\newblock Operator precedence temporal logic and model checking.
\newblock {\em Theor. Comput. Sci.}, 848:47--81, 2020.
\newblock \href {https://doi.org/10.1016/j.tcs.2020.08.034}
  {\path{doi:10.1016/j.tcs.2020.08.034}}.

\bibitem[DBB12]{dubslaff}
Clemens Dubslaff, Christel Baier, and Manuela Berg.
\newblock Model checking probabilistic systems against pushdown specifications.
\newblock {\em Inf. Process. Lett.}, 112(8-9):320--328, 2012.
\newblock \href {https://doi.org/10.1016/j.ipl.2012.01.006}
  {\path{doi:10.1016/j.ipl.2012.01.006}}.

\bibitem[EKL10]{DBLP:journals/siamcomp/EsparzaKL10}
Javier Esparza, Stefan Kiefer, and Michael Luttenberger.
\newblock {Computing the Least Fixed Point of Positive Polynomial Systems}.
\newblock {\em {SIAM} J. Comput.}, 39(6):2282--2335, 2010.
\newblock \href {https://doi.org/10.1137/090749591}
  {\path{doi:10.1137/090749591}}.

\bibitem[EKM04]{esparzaMCPPDA-lics}
Javier Esparza, Anton{\'{\i}}n Kucera, and Richard Mayr.
\newblock Model checking probabilistic pushdown automata.
\newblock In {\em {LICS}}, pages 12--21. {IEEE} Computer Society, 2004.
\newblock \href {https://doi.org/10.1109/LICS.2004.1319596}
  {\path{doi:10.1109/LICS.2004.1319596}}.

\bibitem[EY05]{DBLP:conf/tacas/EtessamiY05}
Kousha Etessami and Mihalis Yannakakis.
\newblock {Algorithmic Verification of Recursive Probabilistic State Machines}.
\newblock In {\em {TACAS}}, volume 3440 of {\em Lecture Notes in Computer
  Science}, pages 253--270. Springer, 2005.
\newblock \href {https://doi.org/10.1007/978-3-540-31980-1\_17}
  {\path{doi:10.1007/978-3-540-31980-1\_17}}.

\bibitem[EY09]{DBLP:journals/jacm/EtessamiY09}
Kousha Etessami and Mihalis Yannakakis.
\newblock Recursive markov chains, stochastic grammars, and monotone systems of
  nonlinear equations.
\newblock {\em J. {ACM}}, 56(1):1:1--1:66, 2009.
\newblock \href {https://doi.org/10.1145/1462153.1462154}
  {\path{doi:10.1145/1462153.1462154}}.

\bibitem[GHNR14]{DBLP:conf/icse/GordonHNR14}
Andrew~D. Gordon, Thomas~A. Henzinger, Aditya~V. Nori, and Sriram~K. Rajamani.
\newblock Probabilistic programming.
\newblock In {\em {FOSE}}, pages 167--181. {ACM}, 2014.
\newblock \href {https://doi.org/10.1145/2593882.2593900}
  {\path{doi:10.1145/2593882.2593900}}.

\bibitem[GMN18]{DBLP:conf/spin/GutsfeldMN18}
Jens~Oliver Gutsfeld, Markus M{\"{u}}ller{-}Olm, and Benedikt Nordhoff.
\newblock {A Branching Time Variant of CaRet}.
\newblock In {\em {SPIN}}, volume 10869 of {\em Lecture Notes in Computer
  Science}, pages 153--170. Springer, 2018.
\newblock \href {https://doi.org/10.1007/978-3-319-94111-0\_9}
  {\path{doi:10.1007/978-3-319-94111-0\_9}}.

\bibitem[GS14]{dippl}
Noah~D Goodman and Andreas Stuhlm\"{u}ller.
\newblock {The Design and Implementation of Probabilistic Programming
  Languages}.
\newblock \url{http://dippl.org}, 2014.
\newblock Accessed: 2023-7-25.

\bibitem[Jae01]{DBLP:journals/amai/Jaeger01}
Manfred Jaeger.
\newblock {Complex Probabilistic Modeling with Recursive Relational {B}ayesian
  Networks}.
\newblock {\em Ann. Math. Artif. Intell.}, 32(1-4):179--220, 2001.
\newblock \href {https://doi.org/10.1023/A:1016713501153}
  {\path{doi:10.1023/A:1016713501153}}.

\bibitem[Jon90]{DBLP:phd/ethos/Jones90}
Claire Jones.
\newblock {\em Probabilistic non-determinism}.
\newblock PhD thesis, University of Edinburgh, {UK}, 1990.
\newblock URL: \url{http://hdl.handle.net/1842/413}.

\bibitem[Kar91]{DBLP:journals/dam/Karp91}
Richard~M. Karp.
\newblock An introduction to randomized algorithms.
\newblock {\em Discret. Appl. Math.}, 34(1-3):165--201, 1991.
\newblock \href {https://doi.org/10.1016/0166-218X(91)90086-C}
  {\path{doi:10.1016/0166-218X(91)90086-C}}.

\bibitem[KEM06]{esparzaMCPPDA}
Anton{\'{\i}}n Kucera, Javier Esparza, and Richard Mayr.
\newblock {Model Checking Probabilistic Pushdown Automata}.
\newblock {\em Log. Methods Comput. Sci.}, 2(1), 2006.
\newblock \href {https://doi.org/10.2168/LMCS-2(1:2)2006}
  {\path{doi:10.2168/LMCS-2(1:2)2006}}.

\bibitem[LMS04]{vpGames}
Christof L{\"{o}}ding, P.~Madhusudan, and Olivier Serre.
\newblock {Visibly Pushdown Games}.
\newblock In {\em {FSTTCS}}, volume 3328 of {\em Lecture Notes in Computer
  Science}, pages 408--420. Springer, 2004.
\newblock \href {https://doi.org/10.1007/978-3-540-30538-5\_34}
  {\path{doi:10.1007/978-3-540-30538-5\_34}}.

\bibitem[MM01]{DBLP:journals/tcs/McIverM01a}
Annabelle McIver and Carroll Morgan.
\newblock Partial correctness for probabilistic demonic programs.
\newblock {\em Theor. Comput. Sci.}, 266(1-2):513--541, 2001.
\newblock \href {https://doi.org/10.1016/S0304-3975(00)00208-5}
  {\path{doi:10.1016/S0304-3975(00)00208-5}}.

\bibitem[OKKM16]{DBLP:conf/lics/OlmedoKKM16}
Federico Olmedo, Benjamin~Lucien Kaminski, Joost{-}Pieter Katoen, and Christoph
  Matheja.
\newblock {Reasoning about Recursive Probabilistic Programs}.
\newblock In {\em {LICS}}, pages 672--681. {ACM}, 2016.
\newblock \href {https://doi.org/10.1145/2933575.2935317}
  {\path{doi:10.1145/2933575.2935317}}.

\bibitem[Pan09]{panangaden2009labelled}
Prakash Panangaden.
\newblock {\em {Labelled Markov Processes}}.
\newblock World Scientific, 2009.

\bibitem[PK00]{DBLP:conf/aaai/PfefferK00}
Avi Pfeffer and Daphne Koller.
\newblock {Semantics and Inference for Recursive Probability Models}.
\newblock In {\em {AAAI/IAAI}}, pages 538--544. {AAAI} Press / The {MIT} Press,
  2000.
\newblock URL: \url{http://www.aaai.org/Library/AAAI/2000/aaai00-082.php}.

\bibitem[Pnu77]{DBLP:conf/focs/Pnueli77}
Amir Pnueli.
\newblock {The Temporal Logic of Programs}.
\newblock In {\em {FOCS}}, pages 46--57. {IEEE} Computer Society, 1977.
\newblock \href {https://doi.org/10.1109/SFCS.1977.32}
  {\path{doi:10.1109/SFCS.1977.32}}.

\bibitem[SEY15]{DBLP:journals/jacm/StewartEY15}
Alistair Stewart, Kousha Etessami, and Mihalis Yannakakis.
\newblock Upper bounds for newton's method on monotone polynomial systems, and
  p-time model checking of probabilistic one-counter automata.
\newblock {\em J. {ACM}}, 62(4):30:1--30:33, 2015.
\newblock \href {https://doi.org/10.1145/2789208} {\path{doi:10.1145/2789208}}.

\bibitem[SG12]{DBLP:conf/starai/StuhlmullerG12}
Andreas Stuhlm{\"{u}}ller and Noah~D. Goodman.
\newblock {A Dynamic Programming Algorithm for Inference in Recursive
  Probabilistic Programs}.
\newblock In {\em StarAI@UAI}, 2012.
\newblock URL:
  \url{https://starai.cs.kuleuven.be/2012/accepted/stuhlmuller.pdf}.

\bibitem[vdMPYW18]{DBLP:journals/corr/abs-1809-10756}
Jan{-}Willem van~de Meent, Brooks Paige, Hongseok Yang, and Frank Wood.
\newblock {An Introduction to Probabilistic Programming}.
\newblock {\em CoRR}, abs/1809.10756, 2018.
\newblock URL: \url{http://arxiv.org/abs/1809.10756}, \href
  {http://arxiv.org/abs/1809.10756} {\path{arXiv:1809.10756}}.

\bibitem[WE07]{DBLP:conf/tacas/WojtczakE07}
Dominik Wojtczak and Kousha Etessami.
\newblock {PReMo: An Analyzer for Probabilistic Recursive Models}.
\newblock In {\em {TACAS}}, volume 4424 of {\em Lecture Notes in Computer
  Science}, pages 66--71. Springer, 2007.
\newblock \href {https://doi.org/10.1007/978-3-540-71209-1\_7}
  {\path{doi:10.1007/978-3-540-71209-1\_7}}.

\bibitem[WGK21]{arxivV2}
Tobias Winkler, Christina Gehnen, and Joost{-}Pieter Katoen.
\newblock {Model Checking Temporal Properties of Recursive Probabilistic
  Programs}.
\newblock {\em CoRR}, abs/2111.03501v2, 2021.
\newblock \href {http://arxiv.org/abs/2111.03501v2}
  {\path{arXiv:2111.03501v2}}.

\bibitem[WGK22]{fossacs}
Tobias Winkler, Christina Gehnen, and Joost{-}Pieter Katoen.
\newblock Model checking temporal properties of recursive probabilistic
  programs.
\newblock In {\em FoSSaCS}, volume 13242 of {\em Lecture Notes in Computer
  Science}, pages 449--469. Springer, 2022.
\newblock \href {https://doi.org/10.1007/978-3-030-99253-8\_23}
  {\path{doi:10.1007/978-3-030-99253-8\_23}}.

\bibitem[WK23a]{DBLP:conf/tacas/WinklerK23}
Tobias Winkler and Joost{-}Pieter Katoen.
\newblock {Certificates for Probabilistic Pushdown Automata via Optimistic
  Value Iteration}.
\newblock In {\em {TACAS} {(2)}}, volume 13994 of {\em Lecture Notes in
  Computer Science}, pages 391--409. Springer, 2023.
\newblock \href {https://doi.org/10.1007/978-3-031-30820-8\_24}
  {\path{doi:10.1007/978-3-031-30820-8\_24}}.

\bibitem[WK23b]{DBLP:conf/lics/WinklerK23}
Tobias Winkler and Joost{-}Pieter Katoen.
\newblock On certificates, expected runtimes, and termination in probabilistic
  pushdown automata.
\newblock In {\em {LICS}}, pages 1--13, 2023.
\newblock \href {https://doi.org/10.1109/LICS56636.2023.10175714}
  {\path{doi:10.1109/LICS56636.2023.10175714}}.

\bibitem[YE05]{DBLP:conf/qest/YannakakisE05}
Mihalis Yannakakis and Kousha Etessami.
\newblock {Checking LTL Properties of Recursive Markov Chains}.
\newblock In {\em {QEST}}, pages 155--165. {IEEE} Computer Society, 2005.
\newblock \href {https://doi.org/10.1109/QEST.2005.8}
  {\path{doi:10.1109/QEST.2005.8}}.

\end{thebibliography}

\end{document}